\documentclass[12pt]{article}
\usepackage{amsthm}
\usepackage{amsmath,bm}
\usepackage{enumerate}
\usepackage{amssymb}
\usepackage{latexsym}
\usepackage{amsfonts}
\usepackage{subfigure}
\usepackage{color}
\usepackage{secdot}
\usepackage{mathrsfs}
\usepackage{epsfig}
\usepackage{natbib}
\usepackage{appendix}
\newtheorem{theorem}{Theorem}[section]
\newtheorem{example}{Example}[section]
\newtheorem{counterexample}{Counterexample}[section]
\newtheorem{corollary}{Corollary}[section]
\newtheorem{definition}{Definition}[section]

\newtheorem{lemma}{Lemma}[section]
\newtheorem{remark}{Remark}[section]

\begin{document}
	\title{Orderings of extremes among dependent extended Weibull random variables}
\author{{\large { {\bf Ramkrishna Jyoti Samanta}$^{a}$\thanks {Email address: akashnilsamanta@gmail.com,}, {\bf Sangita 
 Das}$^{a}$\thanks {Email address: sangitadas118@gmail.com,}~and  {\bf \bf N. Balakrishnan}$^{b}$\thanks {Email address : bala@mcmaster.ca}}} \\
{\small \it $^{a}$Theoretical Statistics and Mathematics Unit, Indian Statistical Institute, Bangalore-560059, India}\\
{\small \it $~^b$Department of Mathematics and Statistics, McMaster University, Hamilton, Ontario L8S 4K1,
				Canada}}
	\date{}
	\maketitle
	\begin{center}
		\noindent{\bf Abstract}
	\end{center}
	 In this work, we consider two sets of dependent variables $\{X_{1},\ldots,X_{n}\}$ and $\{Y_{1},\ldots,Y_{n}\}$, where $X_{i}\sim EW(\alpha_{i},\lambda_{i},k_{i})$ and $Y_{i}\sim EW(\beta_{i},\mu_{i},l_{i})$, for $i=1,\ldots, n$, which are coupled by Archimedean copulas having different generators. Also, let $N_{1}$ and $N_{2}$ be two non-negative integer-valued random variables, independent of $X_{i}'$s and $Y_{i}'$s, respectively. We then establish different inequalities between two extremes, namely, $X_{1:n}$ and $Y_{1:n}$ and $X_{n:n}$ and $Y_{n:n}$, in terms of the usual stochastic, star, Lorenz,  hazard rate, reversed hazard rate and dispersive orders. We also establish some ordering results between $X_{1:{N_{1}}}$ and $Y_{1:{N_{2}}}$ and $X_{{N_{1}}:{N_{1}}}$ and $Y_{{N_{2}}:{N_{2}}}$ in terms of the usual stochastic order. Several examples and counterexamples are presented for illustrating all the results established here. Some of the results here extend the existing results of \cite{barmalzan2020EE}. 

\section{Introduction}
 A $r$-out-of-$n$ system will function if at least $r$ of the $n$ components are functioning. This includes parallel, fail-safe and series systems, corresponding to $r=1$, $r=n-1$, and $r=n$, respectively. We denote the lifetimes of the components by $X_1, \cdots, X_n$, and the corresponding order statistics by $X_{1:n}\le \cdots\le X_{n:n}$. Then, the lifetime of the $r$-out-of-$n$ system is given by $X_{n-r+1:n}$ and so, the theory of order statistics has been used extensively to study the properties of $(n-r+1)$-out-of-$n$ systems. For detailed information on order statistics and their applications, interested readers may refer to \cite{arnold1992first}, \cite{balakrishnan1998} and \cite{Bala_CR_98b}.

 The Weibull distribution has been used in a wide variety of areas, ranging from engineering to finance. 
  Numerous works have been conducted to further explore and analyze properties and different uses of the Weibull distribution. These works further highlight the broad applicability of the Weibull distribution, and its potential for use in many different fields; see, for example, \cite{Mudholkar1996weibull} and \cite{lim2004weibull}.

A flexible family of statistical models is frequently required for data analysis to achieve flexibility while modeling real-life data. Several techniques have been devised to enhance the malleability of a given statistical distribution. One approach is to leverage already well-studied classic distributions, such as gamma, Weibull, and log-normal. Alternatively, one can increase the flexiblity of a distribution by including an additional shape parameter; for instance, the Weibull distribution is generated by taking powers of exponentially distributed random variables. Another popular strategy for achieving this objective, as proposed by 
\cite{marshall1997new}, is to add an extra parameter to any distribution function, resulting in a new family of distributions. To be specific, let $G(x)$ and $\bar{G}(x) = 1-G(x)$ be the distribution and survival functions of a baseline
distribution, respectively. We assume that the distributions have non-negative support. 
Then, it is easy to verify that 
\begin{equation}\label{1}
F(x;\alpha)=\frac{G(x)}{\displaystyle
1-\bar{\alpha}\,\bar{G}(x)},\qquad x,~ \alpha\in(0,\infty),\,
\bar{\alpha}=1-\alpha, 
\end{equation}
and
\begin{equation}\label{2}
F(x;\alpha)=\frac{\alpha\,G(x)}{\displaystyle
1-\bar{\alpha}\,G(x)},\qquad x,~ \alpha\in(0,\infty),\,
\bar{\alpha}=1-\alpha, 
\end{equation}
are both valid cumulative distribution functions. Here, the newly added parameter $\alpha$ is referred to as the tilt parameter. When $G(x)$ has probability density and hazard rate functions as $g(x)$ and $r_G(x)$, respectively, then  the hazard rate function of $F(x;\alpha)$ in \eqref{1} is seen to be
\begin{equation}
     r_{F}(x;\alpha)=\frac{\displaystyle 1}{\displaystyle 1-\bar{\alpha}\,\bar{G}(x)}\,r_{G}(x),\qquad x,~ \alpha\in(0,\infty),\, \bar{\alpha}=1-\alpha.
\end{equation}
Thus, if $r_{G}(x)$ is decreasing (increasing) in $x$, then for $0<\alpha\le 1\,(\alpha\ge 1)$, $r_{F}(x;\alpha)$
is also decreasing (increasing) in $x$. This method has been used by different authors to introduce new extended family of distributions; see, for example,
\cite{mudholkar1993exponentiated}.

 Comparison of two order statistics stochastically has been studied  rather extensively, and especially the comparison of various characteristics of lifetimes of different systems having Weibull components, based on different stochastic orderings. For example, one may see \cite{khaledi2006}, \cite{fang2013weibull}, \cite{fang2014}, \cite{torrado-kochorWeibul2015}, \cite{torradoWeibul2015},  \cite{Zhao2016weibull}, \cite{fang2016ltw}, and the references therein, for stochastic comparisons of series and parallel systems with heterogeneous components with various lifetime distributions.
The majority of existing research on the comparison of series and parallel systems has only considered the case of components that are all independent. However, the operating environment of such technical systems is often subject to a range of factors, such as operating conditions, environmental conditions and the stress factors on the components. For this reason, it would be prudent to take into account the dependence of the lifetimes of components. There are various methods to model this dependence, with the theory of copulas being a popular tool; for example, \cite{nelsen2006introduction} provides a comprehensive account of copulas. Archimedean copulas are a type of multivariate probability distributions used to model the dependence between random variables. They are frequently used in financial applications, such as insurance, risk modeling and portfolio optimization. Many researchers have given consideration to the Archimedean copula due to its flexibility, as it includes the renowned Clayton copula, Ali-Mikhail-Haq copula, and Gumbel-Hougaard copula. Moreover, it also incorporates the independence copula as a special case. As such, results of comparison established under an Archimedean copula for the joint distribution of components' lifespans in a system are general, and would naturally include the corresponding results for the case of independent components.


In this article, we consider the following family of distributions known as extended Weibull family of distributions, with $G(x)=1-e^{-(x\lambda)^{k}},~x,~\lambda,~k>0,$ as the baseline distribution in \eqref{1}. The distribution function of the extended Weibull family is then given by
\begin{equation}\label{4}
   F_{X}(x)=\frac{1-e^{-(x\lambda)^{k}}}{1-\bar{\alpha}e^{-(x\lambda)^k}},~x,~k,~\alpha>0.
\end{equation}
We denote this variable by $X\sim EW(\alpha,\lambda,k),$ where $\alpha$, $\lambda$ and $k$ are respectively known as tilt, scale and shape parameters. In \eqref{4}, if we take $\alpha=1$ and $k=1$, then the extended Weibull family of distributions reduces to the Weibull family of distributions and the extended exponential family of distributions (see, \cite{barmalzan2020EE}), respectively. Similarly, if we take both $\alpha=1$ and $k=1$, the extended Weibull family of distributions reduces to the exponential family of distributions. Now, let us consider two sets of dependent variables $\{X_{1},
\ldots, X_{n}\}$ and $\{Y_{1},
\ldots, Y_{n}\},$  where for $i=1,\ldots,n,$ $X_{i}\sim EW(\alpha_{i},\lambda_{i},k_{i})$ and $Y_{i}\sim EW(\beta_{i},\mu_{i},l_{i})$ are combined with Archimedean (survival) copula having different generators. We then establish here different ordering results between two series and parallel systems, where the systems' components follow extended Weibull family of distributions. The obtained results are based on the usual stochastic, star, Lorenz and dispersive orders. Moreover, considering $\{X_{1},
\ldots, X_{N_{1}}\}$ and $\{Y_{1},
\ldots, Y_{N_{2}}\},$  where $X_{i}\sim EW(\alpha_{i},\lambda_{i},k_{i})$ and $Y_{i}\sim EW(\beta_{i},\mu_{i},l_{i})$ and $N_{1}$ and $N_{2}$ are two random integer-valued random variables independently of $X_{i}'$s and $Y_{i}'$s, respectively, we then compare $X_{1:{N_{1}}}$ and $Y_{1:{N_{2}}}$ and $X_{{N_{1}}:{N_{1}}}$ and $Y_{{N_{2}}:{N_{2}}}$ stochastically.

The rest of this paper is organized as follows. In Section \ref{s1}, we recall some basic stochastic orders and
some important lemmas. The main results are presented in Section \ref{s2}. This part is divided into two subsections. The ordering results between two extreme order statistics are established in Subsection \ref{s21} when the number of variables in the two sets of observations are the same and that the dependent extended Weibull family of distributions have Archimedean (survival) copulas, while in Subsection \ref{s22}, we focus on the case when the two sets have random numbers of variables satisfying the usual stochastic order. In Subsection \ref{s21}, the ordering results are based on the usual stochastic, star, Lorenz,  hazard rate, reversed hazard rate and dispersive orders. Finally, Section \ref{con} presents a brief summary of the work. 

Here, we focus on random variables which are defined on $(0,\infty)$ respresenting lifetimes.
The terms ‘increasing’ and ‘decreasing’ are used in the nonstrict sense. Also, `$\overset{sign}{=}$' is used to
denote that both sides of an equality have the same sign.

\section{Preliminaries}\label{s1}
In this section, we review some important definitions and well-known concepts of stochastic order and majorization which are most pertinant to ensuing discussions. Let  $\boldsymbol{c} =
\left(c_{1},\ldots,c_{n}\right)$ and $\boldsymbol{d} =
\left(d_{1},\ldots,d_{n}\right)$ to be two $n$ dimensional vectors such that $\boldsymbol{c}~,\boldsymbol{d}\in\mathbb{A}.$ Here, $\mathbb{A} \subset \mathbb{R}^{n}$ and $\mathbb{R}^{n}$ is
an $n$-dimensional Euclidean space. Also, consider the order of the elements of the vectors $\boldsymbol{c}$ and $\boldsymbol{d}$ to be $c_{1:n}\leq \ldots \leq c_{n:n}$ and
$d_{1:n}\leq\ldots \leq d_{n:n},$  respectively.
\begin{definition}\label{definition2.2}
	A vector $\boldsymbol{c}$ is said to be
	\begin{itemize}
		\item  majorized by another vector $\boldsymbol{d}$ (denoted
		by $\boldsymbol{c}\preceq^{m} \boldsymbol{d}$) if, for each $l=1,\ldots,n-1$, we have
		$\sum_{i=1}^{l}c_{i:n}\geq \sum_{i=1}^{l}d_{i:n}$ and
		$\sum_{i=1}^{n}c_{i:n}=\sum_{i=1}^{n}d_{i:n};$
		
		\item weakly submajorized by another vector $\boldsymbol{d}$ (denoted
		by $\boldsymbol{c}\preceq_{w} \boldsymbol{d}$) if, for each $l=1,\ldots,n$, we have
		$\sum_{i=l}^{n}c_{i:n}\leq \sum_{i=l}^{n}d_{i:n};$
		
		\item weakly supermajorized by  another vector $\boldsymbol{d},$ denoted
		by $\boldsymbol{c}\preceq^{w} \boldsymbol{d}$, if for each $l=1,\ldots,n$, we have
		$\sum_{i=1}^{l}c_{i:n}\geq \sum_{i=1}^{l}d_{i:n}.$
		
		
		%
	\end{itemize}
\end{definition}
Note that $\boldsymbol{c}\preceq^{m} \boldsymbol{d}$ implies both  $\boldsymbol{c}\preceq_{w} \boldsymbol{d}$ and  $\boldsymbol{c}\preceq^{w} \boldsymbol{d}.$ But, the converse is not always true. For an introduction to majorization order and their applications, are may refer to \cite{Marshall2011}.

Throughout this paper, we are concerned only with  non-negative random variables. Now, we discuss some stochastic orderings. For this purpose, let us suppose $Y$ and $Z$ are two nonn-egative random variables with
probability density functions (PDFs) $f_{Y}$ and $f_{Z}$, cumulative distribution functions (CDFs) $F_{Y}$ and $F_{Z}$, survival functions $\bar
F_{Y}=1-F_{Y}$ and $\bar F_{Z}=1-F_{Z},$ $r_{Y}=f_{Y}/\bar
F_{Y}$,  $ \tilde r_{Y}=f_{Y}/
F_{Y}$ and  $r_{Z}=f_{Z}/
\bar F_{Z}$,  $\tilde r_{Z}=f_{Z}/
F_{Z}$ being the corresponding hazard rate and reversed hazard rate functions, respectively. 


\begin{definition}
	A random variable $Y$ is said to be smaller than $Z$ in the
	\begin{itemize}
 
		\item hazard rate order (denoted by $Y\leq_{hr}Z$) if $r_{Y}(x)\geq r_{Z}(x)$, for all $x;$
		\item reversed hazard rate order (denoted by $Y\leq_{rh}Z$) if $\tilde r_{Y}(x)\leq \tilde r_{Z}(x)$, for all $x;$
		\item usual stochastic order (denoted by $Y\leq_{st}Z$) if
		$\bar F_{Y}(x)\leq\bar F_{Z}(x)$, for all $x;$
  \item dispersive order (denoted by $Y\le_{disp}Z$) if
	\begin{equation*}\label{disp}
		F^{-1}_{Y}(\beta)
		-F^{-1}_{Y}(\alpha)\le F^{-1}_{Z}(\beta)
		-F^{-1}_{Z}(\alpha)\text{ whenever }0<\alpha\leq\beta<1,
	\end{equation*}
where $F^{-1}_{Y}(\cdot)$ and $F^{-1}_{Z}(\cdot)$ are the right-continuous inverses of $F_{Y}(\cdot)$ and $F_{Z}(\cdot),$ respectively;
\item star order (denoted by $X_{1}\leq_{*}X_{2}$) if  $F^{-1}_{Z}F_{Y}(x)$
	is star-shaped in $x,$ that is, ${F^{-1}_{Z}F_{Y}(x)}/{x}$ is increasing in $x\geq 0;$
 \item Lorenz order (denoted by $Y\leq_{Lorenz}Z$) if
	$$\frac{1}{E(Y)}\int_{0}^{F^{-1}_{Y}(u)}x dF_{Y}(x)\geq\frac{1}{E(Z)}\int_{0}^{F^{-1}_{Z}(u)}x dF_{Z}(x),\text{ for all } u\in (0,1].$$
	\end{itemize}
\end{definition}
It is known that the star ordering implies the Lorenz ordering. One may refer to \cite{shaked2007stochastic} for an exhaustive discussion on stochastic orderings. Next, we introduce Schur-convex and
Schur-concave functions.
\begin{lemma}(Theorem 3.A.4 of \citet{Marshall2011}).\label{schur-critetia}
    For an open interval $I\subset R$, a continuously differentiable function $f:I^n\rightarrow R$ is said to be Schur-convex if and only if it is symmetric on $I^n$ and $(x_i-x_j)\big( \frac{\partial f(x)}{\partial x_i}-\frac{\partial f(x)}{\partial x_j}\big)\geq0$ for all $i\neq j$ and $x\in I^n$.
\end{lemma}

	Now, we describe briefly the concept of Archimedean copulas. Let $F$ and $\bar F$ be the joint distribution function and joint survival function of a random vector $\boldsymbol{X}=(X_1,\ldots,X_n)$. Also, suppose there exist some functions $C(\boldsymbol{v}):[0,1]^n\rightarrow [0,1]$ and  
	$\hat {C}(\boldsymbol{v}):[0,1]^n\rightarrow [0,1]$ such that, for all $ x_i,~i\in \mathcal I_n, $ where $\mathcal I_n$ is the index set,
	$$ F(x_1,\ldots,x_n)=C(F_1(x_1),\ldots,F_n(x_n)),$$
	$$\bar{F}(x_1,\ldots,x_n)=\hat{C}(\bar{F_1}(x_1),\ldots,\bar{F_n}(x_n))$$ hold, then $C(\boldsymbol{v})$ and $\hat{C}(\boldsymbol{v})$ are said to be the  copula and survival copula of $\boldsymbol{X}$, respectively. Here, $F_1,\ldots,F_n$ and $\bar{F_1},\ldots,\bar{F_n}$ are the univariate marginal distribution functions and survival functions of the random variables $X_1,\ldots,X_n$, respectively.\\
	Now, suppose $\psi:[0,\infty)\rightarrow[0,1]$ is a non-increasing and continuous function with $\psi(0)=1$ and $\psi(\infty)=0.$ Moreover, suppose $\phi={\psi}^{-1}=\text{sup}\{x\in \mathcal R:\psi(x)>v\}$ is the right continuous inverse. Further, let $\psi$ satisfy the conditions (i) $(-1)^i{\psi}^{(i)}(x)\geq 0,~ i=0,1,\ldots,d-2,$ and  (ii) $(-1)^{d-2}{\psi}^{(d-2)}$ is non-increasing and convex, which imply the generator $\psi$ is $d$-monotone. A copula $C_{\psi}$ is said to be an Archimedean copula if $C_{\psi}$ can be written as $$C_{\psi}(v_1,\ldots,v_n)=\psi({\psi^{-1}(v_1)},\ldots,\psi^{-1}(v_n)),~\text{ for all } v_i\in[0,1],~i\in\mathcal{I}_n.$$ For a detailed discussion on Archimedean copulas, one may refer to \cite{nelsen2006introduction} and \cite{mcneil2009multivariate}.

Next, we present some important lemmas which are essential for the results developed in the following sections.

\begin{lemma}(Lemma 7.1 of \citet{li2015ordering}).\label{Pre-lem2.1f}
	For two $n$-dimensional Archimedean copulas $C_{\psi_1}$ and  $C_{\psi_2}$, if $\phi_2\circ\psi_1$ is super-additive, then $C_{\psi_1}(\boldsymbol{v})\leq C_{\psi_2}(\boldsymbol{v})$, for all $\boldsymbol{v}\in[0,1]^n.$  A function $f$ is said to be super-additive if $ f(x)+f(y)\leq f(x+y),$ for all $x$ and $y$ in the domain of $f.$ Here, $\phi_2$ is the right-continuous inverse of $\psi_2.$
\end{lemma}
\begin{lemma}\label{dec-2}
    Let $f:(0,\infty)\rightarrow (0,\infty)$ be a function given by $f(x)=\dfrac{k\mathrm{e}^{kx}}{1-a\mathrm{e}^{-b^k\mathrm{e}^{kx}}},$ where $0\leq a \leq 1,$ $k,~b>0$. Then, it is increasing in $x,$ for all $x \in (0,\infty)$.
    \begin{proof}
    Taking derivative with respect to $x$, we get $$f^{'}(x)=\dfrac{k^2\cdot\left(\mathrm{e}^{b^k\mathrm{e}^{kx}}-a\cdot\left(b^k\mathrm{e}^{kx}+1\right)\right)\mathrm{e}^{b^k\mathrm{e}^{kx}+kx}}{\left(\mathrm{e}^{b^k\mathrm{e}^{kx}}-a\right)^2}.$$
    Now, as $e^x\geq x+1$ for $x\geq0$, we have $\mathrm{e}^{b^k\mathrm{e}^{kx}}\geq \left(b^k\mathrm{e}^{kx}+1\right)$ which implies $\mathrm{e}^{b^k\mathrm{e}^{kx}}\geq a\cdot\left(b^k\mathrm{e}^{kx}+1\right),$ for $0\leq a\leq 1$. Hence, $f^{'}(x)\geq0$ and therefore $f(x)$ is increasing in $x\in (0,\infty)$.
    \end{proof}
\end{lemma}
\begin{lemma}\label{dec-1}
Let $h:(0,\infty)\rightarrow (0,\infty)$ be a function given by
$h(x)=\dfrac{kx^{k-1}}{1-a\mathrm{e}^{-\left(bx\right)^k}},$ where $a,~b\in (0,\infty)$ and $0<k\leq 1$. Then, it is decreasing in $x$ for all $x\in (0,\infty)$
\begin{proof}
Let $h(x)=kx^{k-1}g(x)$, where $$g(x)=\dfrac{1}{1-a\mathrm{e}^{-\left(bx\right)^k}}.$$
Taking derivative with respect to $x,$ we get
$$g^{'}(x)=-\dfrac{ak\cdot\left(bx\right)^k\mathrm{e}^{\left(bx\right)^k}}{x\cdot\left(\mathrm{e}^{\left(bx\right)^k}-a\right)^2}\leq0.$$
Hence, as $g(x)$ and $kx^{k-1}$ are both non-negative and decreasing functions of $x\in (0,\infty),$ when $0<k\leq 1$, we have $h(x)$ to be decreasing in $x\in (0,\infty)$.
\end{proof}
\end{lemma}

\begin{lemma}\label{dec-3}
Let $m_{1}(x): (0,\infty)\rightarrow(0,\infty)$ 
be a function given by 
    $m_{1}(x)=\mathrm{e}^{\mathrm{e}^{kx}}-a(\mathrm{e}^{kx}+1),$ where $a \in (0,1).$ Then, it is non-negative for $x\in (0,\infty).$
    \begin{proof}
        Differentiating $m_{1}(x)$ with respect to $x$, we get
        $$k\mathrm{e}^{kx}\cdot\left(\mathrm{e}^{\mathrm{e}^{kx}}-a\right)\geq 0.$$
        Because at $x=0,$ we have
        $$\mathrm{e}^{\mathrm{e}^{kx}}-a(\mathrm{e}^{kx}+1)\geq 0,$$ the required result as follows.
    \end{proof}
\end{lemma}

\begin{lemma}\label{dec-4}
    Let $m_{2}(x): (0,\infty)\rightarrow(0,\infty)$ 
be a function, given  by $$m_{2}(x)=\mathrm{e}^{2\mathrm{e}^{kx}}+a\left(\mathrm{e}^{2kx}-3\mathrm{e}^{kx}-2\right)\mathrm{e}^{\mathrm{e}^{kx}}+a^2(\mathrm{e}^{2kx}+3\mathrm{e}^{kx}+1),$$ where $a \in (0,1).$ Then, it is non-negative for  $x \in (0,\infty).$
    \begin{proof}
        Differentiating $m_{2}(x)$ with respect to $x$, we get $$m_{2}^{'}(x)=k\mathrm{e}^{kx}\cdot\left(3a\mathrm{e}^{3\mathrm{e}^{kx}}+2\mathrm{e}^{2\mathrm{e}^{kx}}+\left(-3a\mathrm{e}^{kx}-5a\right)\mathrm{e}^{\mathrm{e}^{kx}}+2a^2\mathrm{e}^{kx}+3a^2\right).$$
We now show that 
\begin{equation}\label{lem2.6}
    k\mathrm{e}^{kx}\cdot\left(3a\mathrm{e}^{3\mathrm{e}^{kx}}+2\mathrm{e}^{2\mathrm{e}^{kx}}+\left(-3a\mathrm{e}^{kx}-5a\right)\mathrm{e}^{\mathrm{e}^{kx}}+2a^2\mathrm{e}^{kx}+3a^2\right)\geq 0.
\end{equation}

For this purpose, let us take $m(x)=\mathrm{e}^{kx}.$ Also, let us set $$g(m)=3a\mathrm{e}^{2m}+2\mathrm{e}^{m}+\left(-3am-5a\right),$$ where $m\geq 1.$ Upon taking partial derivative with respect to $m$, we get
$g^{'}(m)=6ae^{2m}+2e^m-3a$ for $m\geq 1$.
As $g^{''}(m)=12ae^{2m}+2e^m\geq 0$, we have $g^{'}(m)$ to be an increasing function. As the value of $g^{'}(1)\geq 0,$ we obtain the inequality in \eqref{lem2.6}. 
 Further, since at $x=0,$ we have
$$\mathrm{e}^{2\mathrm{e}^{kx}}+a\left(\mathrm{e}^{2kx}-3\mathrm{e}^{kx}-2\right)\mathrm{e}^{\mathrm{e}^{kx}}+a^2(\mathrm{e}^{2kx}+3\mathrm{e}^{kx}+1)\geq 0,$$ the lemma gets established.  
 
    \end{proof}
\end{lemma}

\begin{lemma}\label{lemma2.7}
    Let $m_3(\lambda): (0,\infty)\rightarrow(0,\infty)$ 
be a function given by $$m_{3}(\lambda)=\frac{k\lambda (\lambda x)^{k-1}}{\left(1+(1-a)\mathrm{e}^{\left(x\lambda\right)^k}\right)}
$$ where $0\leq a\leq 1$ and $k\geq 1$. Then, it is convex with respect to $\lambda$
\end{lemma}
\begin{proof}
Taking first and second order partial derivatives of $m_{3}(\lambda)$ with respect to $\lambda,$ we get
$$\frac{\partial m_{3}(\lambda)}{\partial\lambda}=\frac{k\lambda^{k-1}-(1-a)k\lambda^{k-1}e^{-(\lambda t)^k}-(1-a) \lambda^{2k-1} x^k e^{-(\lambda t)^k}}{\left(1+(1-a)\mathrm{e}^{-\left(x\lambda\right)^k}\right)^2}$$ and 
$$\frac{\partial^2 m_{3}(\lambda)}{\partial\lambda^2}=\frac{f_1(\lambda)}{\left(\mathrm{e}^{\left(x\lambda\right)^k}+a-1\right)^3},$$
where $f_1 (\lambda)=\left(k-1\right)e^{2\left(x\lambda\right)^{k}}+a\left(k\left(x\lambda\right)^{2k}-\left(k-1\right)\left(x\lambda\right)^{k}-2\left(k-1\right)\right)e^{\left(x\lambda\right)^{k}}+a^{2}k\left(x\lambda\right)^{2k}+3a^{2}\left(k-1\right)\left(x\lambda\right)^{k}+a^{2}\left(k-1\right).$ To establish the required result, we only need to show that $f(\lambda)\geq 0.$
We first set $(x\lambda)^k=t$ and then observe that $$\left(k-1\right)\left(e^{2\left(x\lambda\right)^{k}}+a\left(\left(x\lambda\right)^{2k}-\left(x\lambda\right)^{k}-2\right)e^{\left(x\lambda\right)^{k}}+a^{2}\left(x\lambda\right)^{2k}+3a^{2}\left(x\lambda\right)^{k}+a^{2}\right)\leq f(\lambda)$$ for $k\geq 1$.
     As $e^t \geq t+1$ for $t\geq 0$, it is enough to show that 
    $$\left(1+t\right)^{2}+a\left(t^{2}-t-2\right)\left(t+1\right)+a^{2}t^{2}+3a^{2}t+a^{2}\geq 0.$$
    It is evident that the above polynomial is greater than $0$ at $t=0$. Now, upon differentiating the above expression with respect to $t$, we get 
    $$3at^2+2(a^2+1)t+(3a^2-3a+2)\geq 0,$$
    which proves that  $$\left(1+t\right)^{2}+a\left(t^{2}-t-2\right)\left(t+1\right)+a^{2}t^{2}+3a^{2}t+a^{2}\geq 0.$$ Hence, we get $f_1(\lambda)\geq 0,$ as required.
    \end{proof}

\section{Main Results}\label{s2}
In this section, we establish different comparison results between two series as well as parallel systems, wherein the systems' components follow extended Weibull distributions with different parameters. The results obtained are in terms of usual stochastic, dispersive and star orders. The modeled parameters are connected with different majorization orders. The main results established here are presented in two subsections; in the first, we consider the sample sizes for the two sets of variables to be equal while in the second they are taken to be random.

\subsection{Ordering results based on equal number of variables}\label{s21}
Let us consider two sets of (equal size) of dependent variables $\{X_{1}\ldots, X_{n}\}$ and $\{Y_{1}\ldots, Y_{n}\},$ where $X_{i}$ and $Y_{i}$ follow dependent extended Weibull distributions having different parameters $\bm{\alpha}=(\alpha_{1},\ldots,\alpha_{n})$, $\bm{\lambda}=(\lambda_{1},\ldots,\lambda_{n}),$ $\bm{k}=(k_{1},\ldots,k_{n})$ and $\bm{\beta}=(\beta_{1},\ldots,\beta_{n})$, $\bm{\mu}=(\mu_{1},\ldots,\mu_{n}),$ $\bm{l}=(l_{1},\ldots,l_{n}),$ respectively.
In the following, we present some results for comparing two extreme order statistics according to their survival functions.
\begin{theorem}\label{th1}
Let $X_i\sim EW(\alpha, \lambda_i,k)  
$ $(i=1,\ldots,n)$ and $Y_i\sim EW(\alpha, \mu_i,k) 
$ $(i=1,\ldots,n)$ have their associated Archimedean survival copulas to be with generators $\psi_1$ and $\psi_2$, respectively. Further, suppose $\phi_2\circ\psi_1$ is super-additive and $\psi_{1}$ is log-concave. Then, for $0<\alpha\leq 1$, we have 
$$(\log{\lambda_1},\ldots,\log{\lambda_n})\succeq_{w}(\log{\mu_1}, \ldots, \log{\mu_n})\Rightarrow Y_{1:n}\succeq_{st}X_{1:n}.$$

\end{theorem}

\begin{proof}
{ The distribution functions of $X_{1:n}$ and $Y_{1:n}$ can be written as }
\begin{equation*}
    F_{X_{1:n}}(x) = 1-\psi_{1}\Big[\sum_{m=1}^{n} \phi_{1}\Big(\frac{\alpha e^{-(x\lambda_m)^k}}{1-\bar{\alpha}e^{-(x\lambda_m)^k}}\Big)\Big]
\end{equation*}
{  and}
\begin{equation*}
    F_{Y_{1:n}}(x) = 1-\psi_{2}\Big[\sum_{m=1}^{n} \phi_{2}\Big(\frac{ \alpha e^{-(x\mu_m)^k}}{1-\bar{\alpha}e^{-(x\mu_m)^k}}\Big)\Big],
\end{equation*}
respectively.
{ Now, from Lemma \ref{Pre-lem2.1f}, super-additivity property of $\phi_{2}\circ\psi_{1}$ yields}
\begin{equation*}
    \psi_{2}\Big[\sum_{m=1}^{n} \phi_{2}\Big(\frac{\alpha e^{-(x\mu_m)^k}}{1-\bar{\alpha}e^{-(x\mu_m)^k}}\Big)\Big]
   \geq \psi_{1}\Big[\sum_{m=1}^{n} \phi_{1}\Big(\frac{\alpha e^{-(x\mu_m)^k}}{1-\bar{\alpha}e^{-(x\mu_m)^k}}\Big)\Big],
\end{equation*}
which is equivalent to stating that
\begin{equation*}
   1-\psi_{2}\Big[\sum_{m=1}^{n} \phi_{2}\Big(\frac{\alpha e^{-(x\mu_m)^k}}{1-\bar{\alpha}e^{-(x\mu_m)^k}}\Big)\Big]
   \leq 1-\psi_{1}\Big[\sum_{m=1}^{n} \phi_{1}\Big(\frac{\alpha e^{-(x\mu_m)^k}}{1-\bar{\alpha}e^{-(x\mu_m)^k}}\Big)\Big].
\end{equation*}
Therefore, {to establish the required result, we only need} to prove that  
\begin{equation*}
    1-\psi_{1}\Big[\sum_{m=1}^{n} \phi_{1}\Big(\frac{\alpha e^{-(x\lambda_m)^k}}{1-\bar{\alpha}e^{-(x\lambda_m)^k}}\Big)\Big]\ge 1-\psi_{1}\Big[\sum_{m=1}^{n} \phi_{1}\Big(\frac{\alpha e^{-(x\mu_m)^k}}{1-\bar{\alpha}e^{-(x\mu_m)^k}}\Big)\Big].
\end{equation*}
{Now, let $\delta(\boldsymbol{e^v})=1-\psi_{1}\Big[\sum_{m=1}^{n} \phi_{1}\Big(\frac{\alpha e^{-(xe^{v_m})^k}}{1-\bar{\alpha}e^{-(xe^{v_m})^k}}\Big)\Big],$ where
$\boldsymbol{e^v}=(e^{v_{1}},\ldots,e^{v_{n}})$ and $(v_{1},\ldots,v_{n})=(\log{\lambda_{1}},\ldots, \log{\lambda_{n}}).$
Due to Theorem A.8 of \cite{Marshall2011}, we just have to -show that $\delta(\boldsymbol{e^v})$ is increasing and Schur-convex in $\boldsymbol{v}.$ 
Taking partial derivative of $\delta(\boldsymbol{e^v})$ with respect to $v_i$, for $i=1,\ldots,n,$ we have}
\begin{equation}\label{th1_1}
 \frac{\partial \delta(\boldsymbol{e^v})}{\partial v_{i}}=\eta(v_i)\Gamma(v_i) \psi_{1}'\Big[\sum_{m=1}^{n}\phi_{1}\Big(\frac{\alpha e^{-(x\lambda_m)^k}}{1-\bar{\alpha}e^{-(x\lambda_{m})^k}}\Big)\Big]\geq 0,
\end{equation}
{ where
$\eta(v_i)= \frac{ke^{v_{i}k}}{1-\bar{\alpha}e^{-(xe^{v_{i}})^k}}$ and $\Gamma(v_i)= \frac{\frac{\alpha e^{-(xe^{v_{i}})^k}}{1-\bar{\alpha}e^{-(xe^{v_{i}})^k}}}{\psi_{1}'\Big(\phi_{1}\Big(\frac{\alpha e^{-(xe^{v_{i}})^k}}{1-\bar{\alpha}e^{-(xe^{v_{i}})^k}}\Big)\Big)},$ for $i={1,\ldots,n}.$
Therefore, from \eqref{th1_1}, we can see that $\delta(\boldsymbol{e^v})$ is increasing in $v_i,$ for $i=1,\ldots,n$.
Now, the derivative of $\Gamma(v_i)$ with respect to $v_i$, is given by
\begin{align*}
    {\Big[\psi_{1}'\Big(\phi_{1}\Big(\frac{\alpha e^{-(xe^{v_{i}})^k}}{1-\bar{\alpha}e^{-(xe^{v_{i}})^k}}\Big)\Big)\Big]^2}\frac{\partial \Gamma(v_{i})}{\partial v_{i}}&=ke^{v_{i}k}\frac{\frac{\alpha x^k e^{-(xe^{v_{i}})^k}}{(1-\bar{\alpha}e^{-(xe^{v_{i}})^k})^2}}{\psi_{1}'\Big(\phi_{1}\Big(\frac{\alpha e^{-(xe^{v_{i}})^k}}{1-\bar{\alpha}e^{-(xe^{v_{i}})^k}}\Big)\Big)}\Big[\Big(\psi_{1}'\Big(\phi_{1}\Big(\frac{\alpha e^{-(xe^{v_{i}})^k}}{1-\bar{\alpha}e^{-(xe^{v_{i}})^k}}\Big)\Big)\Big)^2\nonumber\\
    &-\frac{\alpha e^{-(xe^{v_{i}})^k}}{1-\bar{\alpha}e^{-(xe^{v_{i}})^k}}
    \times{\psi_{1}^{''}\Big(\phi_{1}\Big(\frac{\alpha e^{-(xe^{v_{i}})^k}}{1-\bar{\alpha}e^{-(xe^{v_{i}})^k}}\Big)\Big)\Big]}\leq 0,
\end{align*}
since $\psi_{1}$ is decreasing and log-concave.
This implies that $\Gamma(v_{i})$ is decreasing and non-positive in $v_i,$ for $i=1,\ldots,n$.}
Also, $\eta(v_{i})$ is increasing and non-negative in $v_i,$ from Lemma \ref{dec-2}.
Therefore, $\eta(v_{i})\Gamma(v_{i})$ is decreasing in $v_i,$ for $i=1,\ldots,n$.
Next, we have 
\begin{align*}
    &(v_i-v_j)\Big(\frac{\partial \delta(\boldsymbol{e^v})}{\partial e^{v_{i}}}-\frac{\partial \delta(\boldsymbol{e^v})}{\partial e^{v_{j}}}\Big)\nonumber\\
    &=x^{k}(v_i-v_j)[\eta(v_i)\Gamma(v_i)-\eta(v_j)\Gamma(v_j)]\times\psi_{1}'\Big[\sum_{m=1}^{n}\phi_{1}\Big(\frac{\alpha e^{-(xe^{v_i})^k}}{1-\bar{\alpha}e^{-(xe^{v_i})^k}}\Big)\Big]\nonumber\\&\geq 0. 
\end{align*}
Hence, $\delta(\boldsymbol{e^v})$ is Schur-convex in $\boldsymbol{v}$ from Lemma \ref{schur-critetia}, which completes the proof of the theorem.

\end{proof}

From Theorem 3.1, we can say that if  $\bm{\alpha}$ and $\bm{\beta}$, the shape parameters $\bm{k}$ and $\bm{l}$ are the same and scalar-valued, then under the stated assumptions, we can say that the lifetime $X_{1:n}$ is stochastically less than the lifetime
$Y_{1:n}.$ 

Note in Theorem 3.1 that we have considered the tilt parameter ${\alpha}$ to lie in $(0,1].$ {A natural question that arises is whether under the same condition, we can establish the inequality between two largest order statistics with respect to the usual stochastic order. The following counterexample gives the answer to be negative.
\begin{counterexample}\label{cex1.1}
Let $X_i \sim EW(\alpha,\lambda_i,k)$ ($i=1,2,3$) and $Y_i \sim EW(\alpha,\mu_i,k)$ ($i=1,2,3$).  Set $\alpha=0.55$, $(\lambda_1, \lambda_2, \lambda_3)=(2.14,1.4,1)$ and $(\mu_1, \mu_2, \mu_3)=(0.77, 0.8, 0.8)$. It is then easy to see that $(\log{\lambda_1},\log{\lambda_2}, \log{\lambda_3})\succeq_{w}(\log{\mu_1}, \log{\mu_2},\log{\mu_3})$. 
	Now, suppose we choose the Gumbel-Hougaard copula with parameters $\theta_1=3$ and $\theta_2=0.6$ and $k=1.63$. 
	\begin{figure} \label{fig_1.1}
		\begin{center}
			\includegraphics[height=2.8in]{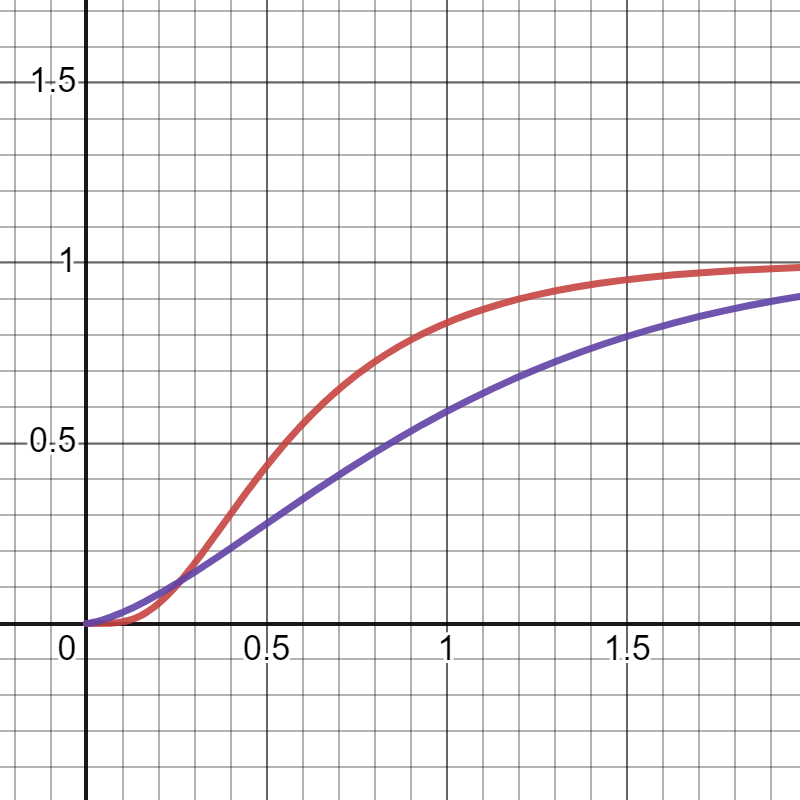}
			\caption{Plots of  ${F}_{X_{3:3}}(x)$ and ${F}_{Y_{3:3}}(x)$ in Counterexample \ref{cex1.1}, where the red line corresponds to ${F}_{X_{3:3}}(x)$ and the blue line  corresponds to ${F}_{Y_{3:3}}(x)$.}
		\end{center}
	\end{figure}
	{Then, Figure 1 presents plots of $F_{X_{3:3}}(x)$ and $F_{Y_{3:3}}(x)$, from which it is evident that when $k \geq 1$, the graph of $F_{X_{3:3}}(x)$ intersects with that of $F_{Y_{3:3}}(x) $ for some $x\geq 0,$ where the two distribution functions are given by}

\begin{align}
F_{X_{3:3}}(x)&=\exp\Bigg\{1-\Bigg(\left(1-\ln\left(\frac{1-\mathrm{e}^{-\left(\lambda_1x\right)^k}}{1-\bar{\alpha}\mathrm{e}^{-\left(\lambda_1x\right)^k}}\right)\right)^\frac{1}{\theta_1}+\left(1-\ln\left(\frac{1-\mathrm{e}^{-\left(\lambda_2x\right)^k}}{1-\Bar{\alpha}\mathrm{e}^{-\left(\lambda_2x\right)^k}}\right)\right)^\frac{1}{\theta_1}\nonumber\\ 
& +\left(1-\ln\left(\frac{1-\mathrm{e}^{-\left(\lambda_3x\right)^k}}{1-\Bar{\alpha}\mathrm{e}^{-\left(\lambda_3x\right)^k}}\right)\right)^\frac{1}{\theta_1}-2\Bigg)^{\theta_1}\Bigg\}
\end{align}

and
\begin{align}
F_{Y_{3:3}}(x)&=\exp\Bigg\{1-\Bigg(\left(1-\ln\left(\frac{1-\mathrm{e}^{-\left(\mu_1x\right)^k}}{1-\Bar{\alpha}\mathrm{e}^{-\left(\mu_1x\right)^k}}\right)\right)^\frac{1}{{\theta_2}}+\left(1-\ln\left(\frac{1-\mathrm{e}^{-\left(\mu_2x\right)^k}}{1-\Bar{\alpha}\mathrm{e}^{-\left(\mu_2x\right)^k}}\right)\right)^\frac{1}{{\theta_2}}\nonumber\\
& +\left(1-\ln\left(\frac{1-\mathrm{e}^{-\left(\mu_3x\right)^k}}{1-\Bar{\alpha}\mathrm{e}^{-\left(\mu_3x\right)^k}}\right)\right)^\frac{1}{\theta_2}-2\Bigg)^{\theta_2}\Bigg\}.
\end{align}
\end{counterexample}
 So, from this counterexample, we show that in order to establish comparisons results between the lifetimes of $X_{n:n}$ and $Y_{n:n},$ we require some other sufficient conditions.}

\begin{theorem}\label{th2}
Let $X_i\sim EW(\alpha, \lambda_i,k)$ $(i=1,\ldots,n)$ and  $Y_i\sim EW(\alpha, \mu_i,k) $ $(i=1,\ldots,n)$ where $0<k\leq1,$  and their associated Archimedean copulas be with generators $\psi_1$ and $\psi_2$, respectively. Also, suppose $\phi_2\circ\psi_1$ is super-additive. Then, for $0<\alpha\leq 1$, we have
$$\boldsymbol{\lambda}\succeq^{w}\boldsymbol{\mu}\Rightarrow X_{n:n}\succeq_{st}Y_{n:n}.$$
\end{theorem}

\begin{proof}
The distribution functions of $X_{n:n}$ and $Y_{n:n}$ can be written as
\begin{equation*}
    F_{X_{n:n}}(x) = \psi_{1}\Big[\sum_{m=1}^{n} \phi_{1}\Big(\frac{1-e^{-(x\lambda_m)^k}}{1-\bar{\alpha}e^{-(x\lambda_m)^k}}\Big)\Big]
\end{equation*}
and
\begin{equation*}
    F_{Y_{n:n}}(x) = \psi_{2}\Big[\sum_{m=1}^{n} \phi_{2}\Big(\frac{1-e^{-(x\mu_m)^k}}{1-\bar{\alpha}e^{-(x\lambda_m)^k}}\Big)\Big],
\end{equation*}
respectively.
Now, from Lemma \ref{Pre-lem2.1f}, the super-additivity of $\phi_{2}$ o $\psi_{1}$ implies that
\begin{equation*}
   \psi_{2}\Big[\sum_{m=1}^{n} \phi_{2}\Big(\frac{1-e^{-(x\mu_m)^k}}{1-\bar{\alpha}e^{-(x\mu_m)^k}}\Big)\Big]
   \geq \psi_{1}\Big[\sum_{m=1}^{n} \phi_{1}\Big(\frac{1-e^{-(x\mu_m)^k}}{1-\bar{\alpha}e^{-(x\mu_m)^k}}\Big)\Big].
\end{equation*}
So, in order to establish the required result, we only need to prove that
\begin{equation*}
    \psi_{1}\Big[\sum_{m=1}^{n} \phi_{1}\Big(\frac{1-e^{-(x\lambda_m)^k}}{1-\bar{\alpha}e^{-(x\lambda_m)^k}}\Big)\Big]\leq \psi_{1}\Big[\sum_{m=1}^{n} \phi_{1}\Big(\frac{1-e^{-(x\mu_m)^k}}{1-\bar{\alpha}e^{-(x\lambda_m)^k}}\Big)\Big].
\end{equation*}
Let us now define $\delta_{1}({ \boldsymbol{\lambda}})=\psi_{1}\Big[\sum_{m=1}^{n} \phi_{1}\Big(\frac{1-e^{-(x\lambda_m)^k}}{1-\bar{\alpha}e^{-(x\lambda_m)^k}}\Big)\Big]$, where ${ \boldsymbol{\lambda}}=(\lambda_1,\ldots,\lambda_n)$.
Due to Theorem A.8 of \cite{Marshall2011}, we just have to show that $\delta_{1}(\boldsymbol{\lambda})$ is increasing and Schur-convcave in $\lambda$. Taking partial derivative of $\delta_1({ \boldsymbol{\lambda}})$ with respect to $\lambda_i,$ for i = $1,\ldots,n,$ we get
\begin{equation*}
 \frac{\partial \delta_{1}(\boldsymbol{\lambda})}{\partial \lambda_i}=\eta_1(\lambda_i)\Gamma_1(\lambda_i)\psi_{1}'\Big[\sum_{m=1}^{n}\phi_{1}\Big(\frac{1- e^{-(x\lambda_m)^k}}{1-\bar{\alpha}e^{-(x\lambda_m)^k}}\Big)\Big]\geq 0,    
\end{equation*}
where
\begin{equation*}
\eta_{1}(\lambda_i)= \frac{k\lambda_i^{k-1}}{1-\bar{\alpha}e^{-(x\lambda_i)^k}}   
\end{equation*}
and 
\begin{equation*}
    \Gamma_{1}(\lambda_i)= \frac{\frac{\alpha e^{-(x\lambda_i)^k}}{1-\bar{\alpha}e^{-(x\lambda_i)^k}}}{\psi_{1}'\Big(\phi_{1}\Big(\frac{1- e^{-(x\lambda_i)^k}}{1-\bar{\alpha}e^{-(x\lambda_i)^k}}\Big)\Big)}.
\end{equation*}
Now, $\eta_{1}(\lambda_i)$ is non-negative and decreasing in $\lambda_i$, for $0 < k\leq 1$, from Lemma \ref{dec-1} and $\Gamma_{1}(\lambda_i)$ is non-positive.
Taking derivative of $\Gamma_{1}(\lambda_i)$ with respect to $\lambda_i$, we get 
{\small\begin{align*}
    \frac{\partial \Gamma_{1}(\lambda_i)}{\partial \lambda_{i}}=&-\Big[\Big(\psi_{1}'\Big(\phi_{1}\Big(\frac{1- e^{-(x\lambda_i)^k}}{1-\bar{\alpha}e^{-(x\lambda_i)^k}}\Big)\Big)\Big)^2+\frac{\alpha e^{-(x\lambda_i)^k}}{1-\bar{\alpha}e^{-(x\lambda_i)^k}}\psi_{1}''\Big(\phi_{1}\Big(\frac{1- e^{-(x\lambda_i)^k}}{1-\bar{\alpha}e^{-(x\lambda_i)^k}}\Big)\Big)\Big]\nonumber\\
    &\times k\lambda_i^{k-1}\frac{\frac{x^k\alpha e^{-(x\lambda_i)^k}}{(1-\bar{\alpha}e^{-(x\lambda_i)^k})^2}}{\psi_{1}'\Big(\phi_{1}\Big(\frac{1- e^{-(x\lambda_i)^k}}{1-\bar{\alpha}e^{-(x\lambda_i)^k}}\Big)\Big)}\times \frac{1}{\Big[\psi_{1}'\Big(\phi_{1}\Big(\frac{1- e^{-(x\lambda_i)^k}}{1-\bar{\alpha}e^{-(x\lambda_i)^k}}\Big)\Big)\Big]^2}\geq 0,
\end{align*}}
 which shows that $\Gamma_{1}(\lambda_i)$ is non-positive and increasing in $\lambda_i,$ for $i={1,2,\ldots,n}$.
Also, $\eta_{1}(\lambda_i)$ is non-negative and decreasing. Hence, $\eta_{1}(\lambda_i) \Gamma_{1}(\lambda_i)$ is increasing in $\lambda_i,$ for $i={1,2,\ldots,n}$.\\
Therefore, for $i\neq j,$
\begin{align*}
    &(\lambda_i-\lambda_j)\Big(\frac{\partial \delta_{1}(\boldsymbol{\lambda})}{\partial \lambda_{i}}-\frac{\partial \delta_{1}(\boldsymbol{\lambda})}{\partial \lambda_{j}}\Big)\nonumber\\
    &=(\lambda_i-\lambda_j)x^k\psi_{1}'\Big[\sum_{m=1}^{n}\phi_{1}\Big(\frac{1- e^{-(x\lambda_i)^k}}{1-\bar{\alpha}e^{-(x\lambda_i)^k}}\Big)\Big][\eta_{1}(\lambda_i)\Gamma_{1}(\lambda_i)-\eta_{1}(\lambda_j)\Gamma_{1}(\lambda_j)]\nonumber\\
    &\leq 0, 
    \end{align*}
which shows that $\delta_1(\boldsymbol{\lambda})$ is Schur-concave in $\boldsymbol{\lambda}$ by { Lemma \ref{schur-critetia}}. Hence, the therorem. 
\end{proof}
\begin{remark}
    In Theorem \ref{th2}, if we take $k=1,$ we simply get the result in Theorem $1$ of \cite{barmalzan2020EE}. 
\end{remark}
\begin{remark} \label{rem1}
	It is important to note that the condition  ``$\phi_{2} \circ \psi_1$    is super-additive''  in Theorem \ref{th2} is quite general and is easy to verify for many
	well-known Archimedean copulas. For example, for  the Gumbel-Hougaard copula with  generator $\psi(t)=e^{1-(1+t)^{\theta}}$ for $\theta \in [1,\infty)$, it is easy to see that $\log \psi(t)=1-(1+t)^{\theta}$ is concave in $t \in [0,1]$. Let us now set $\psi_{1}(t)=e^{1-(1+t)^{\alpha}}$ and $\psi_{2}(t)=e^{1-(1+t)^{\beta}}$. It can then be observed that 
	$\phi_{2} \circ \psi_1(t)=(1+t)^{{\alpha/\beta}-1}.$ Taking  derivative of $\phi_{2} \circ \psi_1(t)$ twice with respect to $t$, it can be seen that 
	$[\phi_{2} \circ \psi_1(t)]^{\prime \prime}=(\frac{\alpha}{\beta})(\frac{\alpha}{\beta}  -1 ) (1+t)^{{\alpha/\beta}-1} \ge 0 $ for $\alpha >\beta>1$, which implies the super-additivity of  $\phi_{2} \circ \psi_1(t)$.
\end{remark}

Now, we present an example that demonstrates that if we consider two parallel systems with their components being mutually dependent with Gumbel-Hougaard copula having  parameters $\theta_1=8.9$ and $\theta_2=3.05$ and following extended Weibull distributions, then under the setup of Theorem \ref{th2}, the survival function of one parallel system is less than that of the other.
\begin{example}\label{ex1}
	Let $X_i \sim EW(\alpha,\lambda_i,k)$ ($i=1,2$) and $Y_i \sim EW(\alpha,\mu_i,k)$ ($i=1,2$).  Set $\alpha=0.6$, $(\lambda_1, \lambda_2)=(0.46,0.5)$ and $(\mu_1, \mu_2)=(1.7,0.43)$. It is then easy to see that $(\mu_1, \mu_2)\stackrel{w}{\preceq} (\lambda_1,\lambda_2)$.  
	Now, suppose we choose the Gumbel-Hougaard copula with parameters $\theta_1=8.9,$ $\theta_2=3.05$ and $k=0.9$. In this case, the distribution functions of $X_{2:2}$ and $Y_{2:2}$ are
	{\fontsize{9}{12}\selectfont\begin{equation*}
	F_{X_{2:2}} (x)=\exp\left\{1-\left(\left[1-\ln\Bigg(\frac{1-e^{-(\lambda_1x)^k}}{1-\bar{\alpha}e^{-(\lambda_1x)^k}}\Bigg)\right]^{1/\theta_1}+\left[1-\ln\Bigg(  \frac{1-e^{-(\lambda_2x)^k}}{1-\bar{\alpha}e^{-(\lambda_2x)^k}}  \Bigg)\right]^{1/\theta_1}-1\right)^{\theta_1} \right\}
	\end{equation*}}
	and
	{\fontsize{9}{12}\selectfont\begin{equation*}
	F_{Y_{2:2}} (x)=\exp\left\{1-\left(\left[1-\ln\Bigg(\frac{1-e^{-(\mu_1x)^k}}{1-\bar{\alpha}e^{-(\mu_1x)^k}}\Bigg)\right]^{1/\theta_2}+\left[1-\ln\left(  \frac{1-e^{-(\mu_2x)^k}}{1-\bar{\alpha}e^{-(\mu_2x)^k}}  \right)\right]^{1/\theta_2}-1\right)^{\theta_2} \right\},
	\end{equation*}}
	respectively. Then, $F_{X_{2:2}}(x) \leq F_{Y_{2:2}}(x)$, for all $x\geq 0$, as already proved in Theorem \ref{th2}.  
	
\end{example}
A natural question that arises here is whether we can extend Theorem \ref{th2} for $k\geq 1.$ The answer to this question is negative as the following counterexample illustrates.

\begin{counterexample}\label{cex1}
	Let $X_i \sim EW(\alpha,\lambda_i,k)$ ($i=1,2$) and $Y_i \sim EW(\alpha,\mu_i,k)$ ($i=1,2$).  Set $\alpha=0.6$, $(\lambda_1, \lambda_2)=(0.46,0.5)$ and $(\mu_1, \mu_2)=(1.7,0.43)$. It is then evident that $(\mu_1, \mu_2)\stackrel{w}{\preceq} (\lambda_1,\lambda_2)$.  
	Now, suppose we choose the Gumbel-Hougaard copula with parameters $\theta_1=8.9,$ $\theta_2=3.05$ and $k=8.06$ which violates the condition stated in Theorem\ref{th2}. 
	\begin{figure} \label{fig_1}
		\begin{center}
			\includegraphics[height=2.8in]{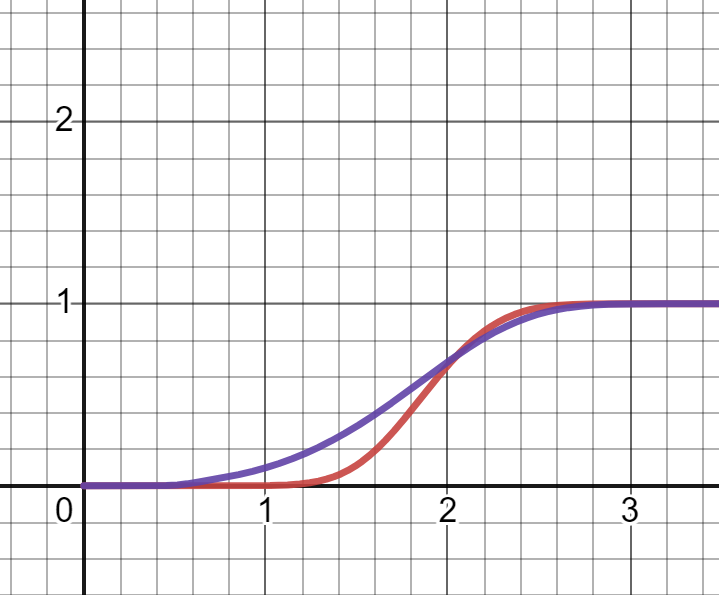}
			\caption{Plots of  ${F}_{X_{2:2}}(x)$ and ${F}_{Y_{2:2}}(x)$ in Counterexample \ref{cex1}, where the red line corresponds to ${F}_{X_{3:3}}(x)$ and blue line corrsponds to ${F}_{Y_{3:3}}(x)$.}
		\end{center}
	\end{figure}
	{Figure 2 plots $F_{X_{2:2}}(x)$ and $F_{Y_{2:2}}(x)$, from which it is evident that when $k \geq 1$, the graph of $F_{X_{2:2}}(x)$ intersects that of $F_{Y_{2:2}}(x) ,$ for some $x\geq 0.$ } 
\end{counterexample}

Now, we establish another result for the case when the shape parameters are connected in majorization order.

\begin{theorem}\label{th4}
Let $X_i\sim EW(\alpha, \lambda,k_i)\hspace{0.1in} (i=1,\ldots,n)$ and $Y_i\sim EW(\alpha, \lambda,l_i)\hspace{0.1in}  (i=1,\ldots,n)$ and the associated Archimedean copula be with generators $\psi_1$ and $\psi_2$, respectively. Further, let $\phi_2\circ\psi_1$ be super-additive and $\alpha t\phi_1^{''}(t)+\phi_1^{'}(t)\geq 0$. Then, for $0<\alpha\leq 1$, we have
$$ \boldsymbol{l}\succeq^{m}\boldsymbol{k} \Rightarrow X_{n:n}\succeq_{st}Y_{n:n}.$$

\end{theorem}
    \begin{proof}
The distribution functions of $X_{n:n}$ and $Y_{n:n}$ can be expressed as 
\begin{equation*}
    F_{X_{n:n}}(x) = \psi_{1}\Big[\sum_{m=1}^{n} \phi_{1}\Big(\frac{1-e^{-(x\lambda)^{k_m}}}{1-\bar{\alpha}e^{-(x\lambda)^{k_m}}}\Big)\Big]
\end{equation*}
and
\begin{equation*}
F_{Y_{n:n}}(x) = \psi_{2}\Big[\sum_{m=1}^{n} \phi_{2}\Big(\frac{1-e^{-(x\lambda)^{l_m}}}{1-\bar{\alpha}e^{-(x\lambda)^{l_m}}}\Big)\Big].
\end{equation*}
Now, from Lemma \ref{Pre-lem2.1f}, the super-additivity of $\phi_{2}$ o $\psi_{1}$ yields
\begin{equation*}
   \psi_{2}\Big[\sum_{m=1}^{n} \phi_{2}\Big(\frac{1-e^{-(x\lambda)^{l_m}}}{1-\bar{\alpha}e^{-(x\lambda)^{l_m}}}\Big)\Big]
   \geq \psi_{1}\Big[\sum_{m=1}^{n} \phi_{1}\Big(\frac{1-e^{-(x\lambda)^{l_m}}}{1-\bar{\alpha}e^{-(x\lambda)^{l_m}}}\Big)\Big].
\end{equation*}
So, to establish the required result, we only need to prove that 
\begin{equation*}
    \psi_{1}\Big[\sum_{m=1}^{n} \phi_{1}\Big(\frac{1-e^{-(x\lambda)^{k_m}}}{1-\bar{\alpha}e^{-(x\lambda)^{k_m}}}\Big)\Big]\leq \psi_{1}\Big[\sum_{m=1}^{n} \phi_{1}\Big(\frac{1-e^{-(x\mu)^{l_m}}}{1-\bar{\alpha}e^{-(x\lambda)^{l_m}}}\Big)\Big].
\end{equation*}
For this purpose, let us define     
    \begin{equation*}
        \delta_3(\boldsymbol{k})=\psi_1\Big[\sum_{m=1}^n \phi_1 \Big(\frac{1-e^{-(\lambda x)^{k_m}}}{1-\bar{\alpha}e^{-(\lambda x)^{k_m}}}\Big)\Big],
    \end{equation*}where $\boldsymbol{k}=(k_1,\ldots,k_n)$. Upon, differentiating $\delta_3(\boldsymbol{k})$  with respect to $k_i$, we get
    \begin{equation*}
        \frac{\partial \delta_3(\boldsymbol{k})}{\partial k_i} = \psi_1^{'}\Big[\sum_{m=1}^n \phi_1 \Big(\frac{1-e^{-(\lambda x)^{k_m}}}{1-\bar{\alpha}e^{-(\lambda x)^{k_m}}}\Big)\Big]  \phi_1^{'} \Big(\frac{1-e^{-(\lambda x)^{k_i}}}{1-\bar{\alpha}e^{-(\lambda x)^{k_i}}}\Big)\frac{(\alpha (\lambda x)^{k_i}\log(\lambda x)e^{-(\lambda x)^{k_i}})}{(1-\bar{\alpha}e^{-(\lambda x)^{k_i}})^2}.
    \end{equation*}
    Let us now define a function $I_3(k_i)$ as
    \begin{equation*}
      I_3(k_i) = \phi_1^{'} \Big(\frac{1-e^{-(\lambda x)^{k_i}}}{1-\bar{\alpha}e^{-(\lambda x)^{k_i}}}\Big)\frac{(\alpha (\lambda x)^{k_i}\log(\lambda x)e^{-(\lambda x)^{k_i}})}{(1-\bar{\alpha}e^{-(\lambda x)^{k_i}})^2}
    \end{equation*}
 which, upon differentiating with respect to $k_i$, yields  
    \begin{eqnarray*}
        \frac{\partial I_3(k_i)}{\partial k_i}
        &=& \phi_1^{''} \Big(\frac{1-e^{-(\lambda x)^{k_i}}}{1-\bar{\alpha}e^{-(\lambda x)^{k_i}}}\Big)\bigg(\frac{(\alpha (\lambda x)^{k_i}\log(\lambda x)e^{(\lambda x)^{k_i}})}{(e^{(\lambda x)^{k_i}}-\bar{\alpha})^2}\bigg)^2 \\
        &-& \phi_1^{'} \Big(\frac{1-e^{-(\lambda x)^{k_i}}}{1-\bar{\alpha}e^{-(\lambda x)^{k_i}}}\Big)\frac{\alpha (\lambda x)^{k_i}\log(\lambda x)^2 e^{(\lambda x)^{k_i}\big(((\lambda x)^{k_i}-1)e^{(\lambda x)^{k_i}}+\bar{\alpha}(\lambda x)^{k_i}+\bar{\alpha}\big)}}{(e^{(\lambda x)^{k_i}}-\bar{\alpha})^3}\\
        &=& \phi_1^{''} \Big(\frac{1-e^{-(\lambda x)^{k_i}}}{1-\bar{\alpha}e^{-(\lambda x)^{k_i}}}\Big)\bigg(\frac{(\alpha (\lambda x)^{k_i}\log(\lambda x)e^{(\lambda x)^{k_i}})}{(e^{(\lambda x)^{k_i}}-\bar{\alpha})^2}\bigg)^2 \\
        &&
        - \phi_1^{'} \Big(\frac{1-e^{-(\lambda x)^{k_i}}}{1-\bar{\alpha}e^{-(\lambda x)^{k_i}}}\Big)\frac{\alpha (\lambda x)^{k_i}\log(\lambda x)^2 e^{(\lambda x)^{k_i}\big(((\lambda x)^{k_i}-1)e^{(\lambda x)^{k_i}}+\bar{\alpha}(\lambda x)^{k_i}+\bar{\alpha}\big)}}{(e^{(\lambda x)^{k_i}}-\bar{\alpha})^3}\\
        &=& \phi_1^{''} \Big(\frac{1-e^{-(\lambda x)^{k_i}}}{1-\bar{\alpha}e^{-(\lambda x)^{k_i}}}\Big)\Big(\frac{1-e^{-(\lambda x)^{k_i}}}{1-\bar{\alpha}e^{-(\lambda x)^{k_i}}}\Big)\\
        && - \phi_1^{'} \Big(\frac{1-e^{-(\lambda x)^{k_i}}}{1-\bar{\alpha}e^{-(\lambda x)^{k_i}}}\Big)\frac{\big(((\lambda x)^{k_i}-1)e^{(\lambda x)^{k_i}}+\bar{\alpha}(\lambda x)^{k_i}+\bar{\alpha}\big)(1-e^{-(\lambda x)^{k_i}})}{\alpha (\lambda x)^{k_i}}.
    \end{eqnarray*}
    Now, since for $ 0\leq a\leq 1$ and $x\geq0$, $$\frac{(1-x)e^x+ax+a}{x}(1-e^{-x})\geq -1,$$ 
    we obtain
    $$\frac{\partial I_3(k_i)}{\partial k_i}\geq \phi_1^{''} \Big(\frac{1-e^{-(\lambda x)^{k_i}}}{1-\bar{\alpha}e^{-(\lambda x)^{k_i}}}\Big)\Big(\frac{1-e^{-(\lambda x)^{k_i}}}{1-\bar{\alpha}e^{-(\lambda x)^{k_i}}}\Big) + \phi_1^{'} \Big(\frac{1-e^{-(\lambda x)^{k_i}}}{1-\bar{\alpha}e^{-(\lambda x)^{k_i}}}\Big)\frac{1}{\alpha}\geq 0$$
    as $\alpha t\phi_1^{''}(t)+\phi_1^{'}(t)\geq 0$,  and $I_3(k_i)$ is increasing in $k_i,$ for $i=1,\ldots,n.$\\
    Now, for $i\neq j,$
\begin{align*}
    &(k_i-k_j)\Big(\frac{\partial \delta_3(\boldsymbol{k})}{\partial k_{i}}-\frac{\partial \delta_{3}(\boldsymbol{k})}{\partial k_{j}}\Big)\nonumber\\
    &=(k_i-k_j)\psi_{1}'\Big[\sum_{m=1}^{n}\phi_{1}\Big(\frac{1- e^{-(x\lambda)^{k_m}}}{1-\bar{\alpha}e^{-(x\lambda)^{k_m}}}\Big)\Big][I_3(k_i)-I_3(k_j)]\nonumber\\
    &\leq 0, 
    \end{align*}
which implies $\delta_3(\boldsymbol{k})$ is Schur-concave in $\boldsymbol{k}$ Lemma \ref{schur-critetia}. This completes the proof the theorem.
\end{proof}
\begin{remark}\label{rem2}
	It is useful to observe that the condition  ''$\phi_{2} \circ \psi_1$ is super-additive''  in Theorem \ref{th6} is quite general and is easy to verify for many
	well-known Archimedean copulas. For example, we consider the copula
	$\phi(t)=e^{\frac{\theta}{t}}-e^{\theta}$, for $t\geq 0$, that satisfies the relation $$t\phi^{''}(t)+2\phi^{'}(t)\geq0,$$ where $t\geq0$ and the inverse of $\phi(t)$ is $$\psi(t)=\frac{\theta}{\ln(t+e^{\theta})}.$$ Suppose we have two such copulas, but with parameters $\alpha$ and $\beta$, and we want to find the condition for $\phi_{2} \circ \psi_1$ to be super-additive. Taking double-derivative of $\phi_{2} \circ \psi_1$ with respect to $x$, we get that
	$$(\phi_{2} \circ \psi_1)^{''}(x)=\left(\frac{\beta}{\alpha}\right)\left(\frac{\beta}{\alpha}-1\right)(t+e^{\theta_1})^{\frac{\beta}{\alpha}-2}.$$
\end{remark}

For illustrating the result in Theorem \ref{th4}, let us consider the following example.

\begin{example}\label{ex2}
	Let $X_i \sim EW(\alpha,\lambda,k_i)$ ($i=1,2,3$) and $Y_i \sim EW(\alpha,\lambda,l_i)$ ($i=1,2,3$).  Set $\alpha=0.5$, $\lambda= 4.83$, $(k_1, k_2, k_3)=(3, 0.5, 1)$ and $(l_1, l_2, l_3)=(2, 1.5, 1)$. It is then easy to see that $ (k_1,k_2,k_3) \stackrel{m}{\preceq} (l_1, l_2, l_3)$.  
	Now, suppose we choose the copula $\phi(t)=e^{\frac{\theta}{t}}-e^{\theta}$, for $t\geq 0$, that satisfies the relation  $$\alpha t\phi_1^{''}(t)+\phi_1^{'}(t)\geq0,$$ where $t\geq0$, as $\alpha=0.5$ and the parameters $\theta_1=2.2$ and $\theta_2=2.45$ ensures the super-additivity of $\phi_{2} \circ \psi_1$. The distribution functions of $X_{3:3}$ and $Y_{3:3}$ are,   respectively,
	\begin{equation*}
	F_{X_{3:3}} (x)=\frac{\theta_{1}}{\ln\left(e^{\theta_1\frac{1-\bar{\alpha}{\mathrm{e}}^{-(\lambda x)^{k_{1}}}}{1-{\mathrm{e}}^{-(\lambda x)^{k_{1}}}}}+e^{\theta_1\frac{1-\bar{\alpha}{\mathrm{e}}^{-(\lambda x)^{k_{2}}}}{1-{\mathrm{e}}^{-(\lambda x)^{k_{2}}}}}+e^{\theta_1\frac{1-\bar{\alpha}{\mathrm{e}}^{-(\lambda x)^{k_{3}}}}{1-{\mathrm{e}}^{-(\lambda x)^{k_{3}}}}}-2\mathrm{e}^{\theta_{1}}\right)}
	\end{equation*}
	and
	\begin{equation*}
	F_{Y_{3:3}} (x)=\frac{\theta_{2}}{\ln\left(e^{\theta_2\frac{1-\bar{\alpha}{\mathrm{e}}^{-(\lambda x)^{l_{1}}}}{1-{\mathrm{e}}^{-(\lambda x)^{l_{1}}}}}+e^{\theta_2\frac{1-\bar{\alpha}{\mathrm{e}}^{-(\lambda x)^{l_{2}}}}{1-{\mathrm{e}}^{-(\lambda x)^{l_{2}}}}}+e^{\theta_2\frac{1-\bar{\alpha}{\mathrm{e}}^{-(\lambda x)^{l_{3}}}}{1-{\mathrm{e}}^{-(\lambda x)^{l_{3}}}}}-2\mathrm{e}^{\theta_{2}}\right)}.
	\end{equation*}
	
 Then, $F_{X_{3:3}}(x) \leq F_{Y_{3:3}}(x),$ for all $x\geq 0,$ as established in Theorem \ref{th4}.
	
\end{example}
It is useful to observe that the condition ``$\phi_{2} \circ \psi_1$ is super-additive" provides the copula with generator $\psi_2$ to be more positively dependent than the copula with generator $\psi_1.$
In Theorem \ref{th4}, we have considered $\phi_{2} \circ \psi_1$ to be super-additive which is important to establish the inequality between the survival functions of $X_{n:n}$ and $Y_{n:n}$ when the parameters $\bm{l}$ and $\bm{k}$ are comparable in terms of majorization order. We now present a counterexample which allows us to show that if the condition is violated, then the theorem does not hold.  
\begin{counterexample}\label{cex2}
	Let $X_i \sim EW(\alpha,\lambda,k_i)$ ($i=1,2,3$) and $Y_i \sim EW(\alpha,\mu,l_i)$ ($i=1,2,3$).  Set $\alpha=0.5$, $\lambda= 4.83$, $(k_1, k_2, k_3)=(0.5, 1 , 3)$ and $(l_1, l_2, l_3)=(1, 1.5, 2)$. It is easy to see  $ (k_1,k_2,k_3) \stackrel{m}{\preceq} (l_1, l_2, l_3)$.  
 Now, suppose we choose the copula $\phi(t)=e^{\frac{\theta}{t}}-e^{\theta}$, for $t\geq 0$, that satisfies $$\alpha t\phi^{''}(t)+\phi^{'}(t)\geq0,$$ where $t\geq0$, as $\alpha=0.5$ and the parameters $\theta_1=2.48$ and $\theta_2=2.24$ violate the condition of super-additivity of $\phi_{2} \circ \psi_1$.
	\begin{figure} \label{fig2}
		\begin{center}
			\includegraphics[height=2.8in]{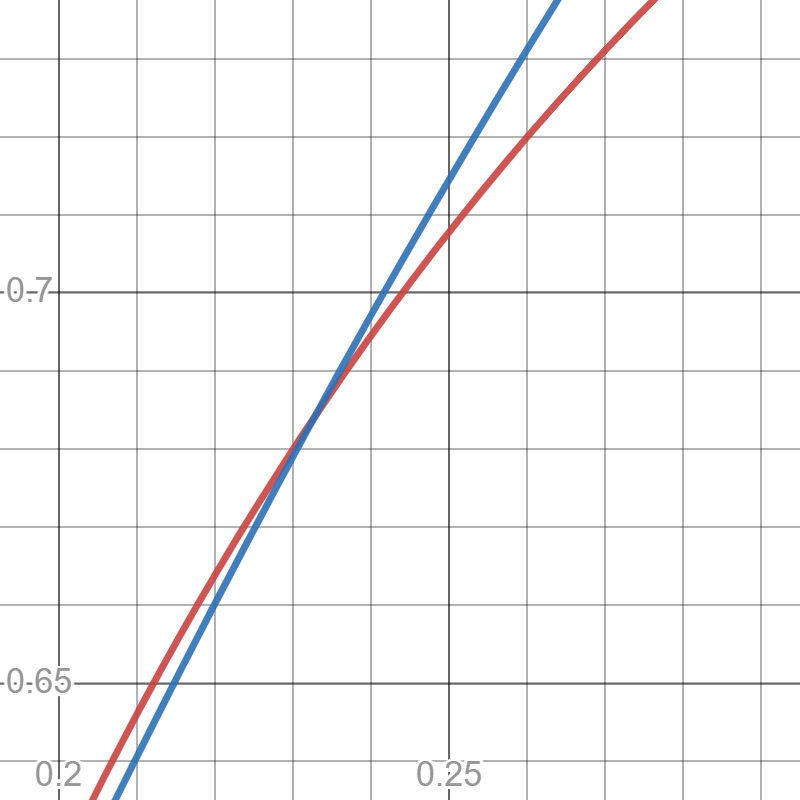}
			\caption{Plots of  ${F}_{X_{3:3}}(x)$ and ${F}_{Y_{3:3}}(x)$ in Counterexample \ref{cex2}. Here, the red line corresponds to ${F}_{X_{3:3}}(x)$ and the blue line corresponds to ${F}_{Y_{3:3}}(x)$.}
		\end{center}
	\end{figure}
 Figure 3 plots $F_{X_{3:3}}(x)$ and $F_{Y_{3:3}}(x)$ from it, is evident that when the condition of super-additivity of $\phi_{2} \circ \psi_1$ is violated in Theorem \ref{th4}, $F_{X_{3:3}}(x)$ is greater than $ F_{Y_{3:3}}(x) $ for some $x\geq 0$.
\end{counterexample}
In Theorem \ref{th4}, if we replace the condition $\alpha t\phi_1^{''}(t)+\phi_1^{'}(t)\geq 0$ by $t\phi_1^{''}(t)+\phi_1^{'}(t)\geq 0$, then we can also compare the smallest order statistics $X_{1:n}$ and $Y_{1:n}$ with respect to the usual stochastic order under the same conditions as in Theorem \ref{th4}. 
\begin{theorem}\label{th5}
Let $X_i\sim EW(\alpha, \lambda,k_i)\hspace{0.1in}(i=1,\ldots,n)$ and $Y_i\sim EW(\alpha, \mu,l_i)\hspace{0.1in}(i=1,\ldots,n)$ and the associated Archimedean survival copulas be with generators $\psi_1$ and $\psi_2$, respectively. Also, let $\phi_2\circ\psi_1$ be super-additive and $t\phi_1^{''}(t)+\phi_1^{'}(t)\geq 0$. Then, for $0<\alpha\leq 1$, we have
$$ \boldsymbol{l}\succeq^{m}\boldsymbol{k} \Rightarrow Y_{1:n}\succeq_{st}X_{1:n}.$$
\end{theorem}
    \begin{proof}
The distribution functions of $X_{1:n}$ and $Y_{1:n}$ can be written as 
\begin{equation*}
    F_{X_{1:n}}(x) = 1-\psi_{1}\Big[\sum_{m=1}^{n} \phi_{1}\Big(\frac{\alpha e^{-(x\lambda)^{k_m}}}{1-\bar{\alpha}e^{-(x\lambda)^{k_m}}}\Big)\Big]
\end{equation*}
and
\begin{equation*}
    F_{Y_{1:n}}(x) = 1-\psi_{2}\left[\sum_{m=1}^{n} \phi_{2}\left(\frac{\alpha e^{-(x\lambda)^{l_m}}}{1-\bar{\alpha}e^{-(x\lambda)^{l_m}}}\right)\right],
\end{equation*}
respectively.
Now from Lemma \ref{Pre-lem2.1f}, the super-additivity of $\phi_{2}$ o $\psi_{1}$ provides
\begin{equation*}
   \psi_{2}\Big[\sum_{m=1}^{n} \phi_{2}\Big(\frac{\alpha e^{-(x\lambda)^{l_m}}}{1-\bar{\alpha}e^{-(x\lambda)^{l_m}}}\Big)\Big]
   \geq \psi_{1}\Big[\sum_{m=1}^{n} \phi_{1}\Big(\frac{\alpha e^{-(x\lambda)^{l_m}}}{1-\bar{\alpha}e^{-(x\lambda)^{l_m}}}\Big)\Big].
\end{equation*}
Hence, to establish the required result, we only need to prove that
\begin{equation*}
    \psi_{1}\Big[\sum_{m=1}^{n} \phi_{1}\Big(\frac{\alpha e^{-(x\lambda)^{k_m}}}{1-\bar{\alpha}e^{-(x\lambda)^{k_m}}}\Big)\Big]\leq \psi_{1}\Big[\sum_{m=1}^{n} \phi_{1}\Big(\frac{\alpha e^{-(x\mu)^{l_m}}}{1-\bar{\alpha}e^{-(x\lambda)^{l_m}}}\Big)\Big].
\end{equation*}
Let us define     
    \begin{eqnarray*}
        \delta_4(\boldsymbol{k})&=& 1-\psi_1\Big[\sum_{m=1}^n \phi_1 \Big(\frac{\alpha e^{-(\lambda x)^{k_m}}}{1-\bar{\alpha}e^{-(\lambda x)^{k_m}}}\Big)\Big],
    \end{eqnarray*}where $\boldsymbol{k}=(k_1,\ldots,k_n)$. Upon taking partial-derivative of $\delta_4(\boldsymbol{k})$ with respect to $k_i$, we get
    \begin{eqnarray*}
        \frac{\partial \delta_4(\boldsymbol{k})}{\partial k_i} = -\psi_1^{'}\Big[\sum_{i=1}^n \phi_1 \Big(\frac{\alpha e^{-(\lambda x)^{k_i}}}{1-\bar{\alpha}e^{-(\lambda x)^{k_i}}}\Big)\Big]  \phi_1^{'} \Big(\frac{\alpha e^{-(\lambda x)^{k_i}}}{1-\bar{\alpha}e^{-(\lambda x)^{k_i}}}\Big)\frac{((-1)\alpha (\lambda x)^{k_i}\log(\lambda x)e^{-(\lambda x)^{k_i}})}{(1-\bar{\alpha}e^{-(\lambda x)^{k_i}})^2}.
    \end{eqnarray*}
    Next, let us define a function $I_4(k_i)$ as
    \begin{eqnarray*}
      I_4(k_i) = -\phi_1^{'} \Big(\frac{\alpha e^{-(\lambda x)^{k_i}}}{1-\bar{\alpha}e^{-(\lambda x)^{k_i}}}\Big)\frac{((-1)\alpha (\lambda x)^{k_i}\log(\lambda x)e^{-(\lambda x)^{k_i}})}{(1-\bar{\alpha}e^{-(\lambda x)^{k_i}})^2},
    \end{eqnarray*}
 which upon taking partial-derivative with respect to $k_i$ yields  
 \begin{eqnarray*}
        \frac{\partial I_4(k_i)}{\partial k_i}
       &=& -\phi_1^{''} \Big(\frac{\alpha e^{-(\lambda x)^{k_i}}}{1-\bar{\alpha}e^{-(\lambda x)^{k_i}}}\Big)\bigg(\frac{(\alpha (\lambda x)^{k_i}\log(\lambda x)e^{(\lambda x)^{k_i}})}{(e^{(\lambda x)^{k_i}}-\bar{\alpha})^2}\bigg)^2 \\
        &&+ \phi_1^{'} \Big(\frac{\alpha e^{-(\lambda x)^{k_i}}}{1-\bar{\alpha}e^{-(\lambda x)^{k_i}}}\Big)\frac{\alpha (\lambda x)^{k_i}\log(\lambda x)^2 e^{(\lambda x)^{k_i}}\big(((\lambda x)^{k_i}-1)e^{(\lambda x)^{k_i}}+\bar{\alpha}(\lambda x)^{k_i}+\bar{\alpha}\big)}{(e^{(\lambda x)^{k_i}}-\bar{\alpha})^3}\\&=& -\phi_1^{''} \Big(\frac{\alpha e^{-(\lambda x)^{k_i}}}{1-\bar{\alpha}e^{-(\lambda x)^{k_i}}}\Big)\bigg(\frac{(\alpha (\lambda x)^{k_i}\log(\lambda x)e^{(\lambda x)^{k_i}})}{(e^{(\lambda x)^{k_i}}-\bar{\alpha})^2}\bigg)^2 \\&&
        + \phi_1^{'} \Big(\frac{\alpha e^{-(\lambda x)^{k_i}}}{1-\bar{\alpha}e^{-(\lambda x)^{k_i}}}\Big)\frac{\alpha (\lambda x)^{k_i}\log^2(\lambda x) e^{(\lambda x)^{k_i}}\big(((\lambda x)^{k_i}-1)e^{(\lambda x)^{k_i}}+\bar{\alpha}(\lambda x)^{k_i}+\bar{\alpha}\big)}{(e^{(\lambda x)^{k_i}}-\bar{\alpha})^3}\\&=& -\phi_1^{''} \Big(\frac{\alpha e^{-(\lambda x)^{k_i}}}{1-\bar{\alpha}e^{-(\lambda x)^{k_i}}}\Big)\Big(\frac{\alpha e^{-(\lambda x)^{k_i}}}{1-\bar{\alpha}e^{-(\lambda x)^{k_i}}}\Big)\\&& + \phi_1^{'} \Big(\frac{\alpha e^{-(\lambda x)^{k_i}}}{1-\bar{\alpha}e^{-(\lambda x)^{k_i}}}\Big)\frac{\big(((\lambda x)^{k_i}-1)e^{(\lambda x)^{k_i}}+\bar{\alpha}(\lambda x)^{k_i}+\bar{\alpha}\big)(e^{-(\lambda x)^{k_i}})}{(\lambda x)^{k_i}}.
    \end{eqnarray*}
    Now, since for $ 0\leq a\leq 1$ and $x\geq0$, $$\frac{(x-1)e^x+ax+a}{x}e^{-x}\leq 1,$$
    we obtain
    $$\frac{\partial I_4(k_i)}{\partial k_i}\geq -\phi_1^{''} \Big(\frac{\alpha e^{-(\lambda x)^{k_i}}}{1-\bar{\alpha}e^{-(\lambda x)^{k_i}}}\Big)\Big(\frac{\alpha e^{-(\lambda x)^{k_i}}}{1-\bar{\alpha}e^{-(\lambda x)^{k_i}}}\Big) + \phi_1^{'} \Big(\frac{1-e^{-(\lambda x)^{k_i}}}{1-\bar{\alpha}e^{-(\lambda x)^{k_i}}}\Big),$$
   and since $t\phi_1^{''}(t)+\phi_1^{'}(t)\geq 0$, it shows $I_4(k_i)$ is increasing in $k_i$, for $i=1,\ldots,n$.\\
    Finally, for $i\neq j,$
\begin{align*}
    &(k_i-k_j)\Big(\frac{\partial \delta_4(\boldsymbol{k})}{\partial k_{i}}-\frac{\partial \delta_4(\boldsymbol{k})}{\partial k_{j}}\Big)\nonumber\\
    &=(k_i-k_j)\psi_{1}'\Big[\sum_{m=1}^{n}\phi_{1}\Big(\frac{1- e^{-(x\lambda)^{k_m}}}{1-\bar{\alpha}e^{-(x\lambda)^{k_m}}}\Big)\Big][I_4(k_i)-I_4(k_j)]\nonumber\\
    &\leq 0, 
    \end{align*}
which implies $\delta_4(\boldsymbol{k})$ is Schur-concave in $\boldsymbol{k}$ by Lemma \ref{schur-critetia}. This completes the proof of the theorem. 
\end{proof}

In all the previous theorems, we have developed results concerning the usual stochastic order between two extremes, where the tilt parameters for both sets of variables are the same and scalar-valued. Next, we prove another result for comparing two parallel systems containing $n$ number of dependent  components following extended Weibull distribution wherein the dependency is modelled by Archimedean  copulas having different generators and the tilt parameters are connected in weakly sub-majorization order. To establish the following theorem, we need $\phi_2\circ\psi_1$ to be super-additive and 
$t\phi_1^{''}(t)+2\phi_1^{'}(t)\geq 0,$ where $
\phi_{1}$ is the inverse of $\psi_1$.
\begin{theorem}\label{th6}
Let $X_i\sim EW(\alpha_i,\lambda,k)\hspace{0.1in} (i=1,\ldots,n)$ and $Y_i\sim EW(\beta_i, \lambda, k)\hspace{0.1in}(i=1,\ldots,n)$ and the associated Archimedean copulas be with generators $\psi_1$ and $\psi_2$, respectively. Also, let $\phi_2\circ\psi_1$ be super-additive and 
\begin{equation*}
t\phi_1^{''}(t)+2\phi_1^{'}(t)\geq 0.
\end{equation*} Then, we have
$$ \boldsymbol{\alpha}\succeq_{w}\boldsymbol{\beta} \Rightarrow X_{n:n}\succeq_{st}Y_{n:n}.$$
\end{theorem}

\begin{proof}

The distribution functions of $X_{n:n}$ and $Y_{n:n}$ can be written as 
\begin{equation*}
    F_{X_{n:n}}(x) = \psi_{1}\Big[\sum_{m=1}^{n} \phi_{1}\Big(\frac{1-e^{-(x\lambda)^{k}}}{1-\bar{\alpha_m}e^{-(x\lambda)^{k}}}\Big)\Big]
\end{equation*}
and
\begin{equation*}
    F_{Y_{n:n}}(x) = \psi_{2}\Big[\sum_{m=1}^{n} \phi_{2}\Big(\frac{1-e^{-(x\lambda)^{k}}}{1-\bar{\beta_m}e^{-(x\lambda)^{k}}}\Big)\Big],
\end{equation*}
respectively.
Now from Lemma \ref{Pre-lem2.1f}, the super-additivity of $\phi_{2}$ o $\psi_{1}$ provides
\begin{equation*}
   \psi_{2}\Big[\sum_{m=1}^{n} \phi_{2}\Big(\frac{1-e^{-(x\lambda)^{k}}}{1-\bar{\beta_m}e^{-(x\lambda)^{k}}}\Big)\Big]
   \geq \psi_{1}\Big[\sum_{m=1}^{n} \phi_{1}\Big(\frac{1-e^{-(x\lambda)^{k}}}{1-\bar{\beta_m}e^{-(x\lambda)^{k}}}\Big)\Big].
\end{equation*}
Hence, to establish the required result, we only need to prove  that
\begin{equation*}
    \psi_{1}\Big[\sum_{m=1}^{n} \phi_{1}\Big(\frac{1-e^{-(x\lambda)^{k}}}{1-\bar{\beta_m}e^{-(x\lambda)^{k}}}\Big)\Big]\geq \psi_{1}\Big[\sum_{m=1}^{n} \phi_{1}\Big(\frac{1-e^{-(x\lambda)^{k}}}{1-\bar{\alpha_m}e^{-(x\lambda)^{k}}}\Big)\Big].
\end{equation*}
Let us define     
    \begin{eqnarray*}
        \delta_5(\boldsymbol{\alpha})&=&\psi_1\Big[\sum_{m=1}^n \phi_1 \Big(\frac{1-e^{-(\lambda x)^{k}}}{1-\bar{\alpha_m}e^{-(\lambda x)^{k}}}\Big)\Big].
    \end{eqnarray*} Now, differentiating $\delta_5(\boldsymbol{\alpha})$ with respect to $\alpha_i$, we get 
\begin{eqnarray*}
\frac{\partial \delta_5(\boldsymbol{\alpha})}{\partial \alpha_i}&=&\psi_1^{'}\Big[\sum_{i=1}^n \phi_1 \Big(\frac{1-e^{-(\lambda x)^{k}}}{1-\bar{\alpha_i}e^{-(\lambda x)^{k}}}\Big)\Big]\phi_1^{'} \Big(\frac{1-e^{-(\lambda x)^{k}}}{1-\bar{\alpha_i}e^{-(\lambda x)^{k}}}\Big)\frac{(-1)(1-e^{-(\lambda x)^k})e^{-(\lambda x)^k}}{(1-\bar{\alpha_i}e^{-(\lambda x)^k} )^2} \leq 0.
\end{eqnarray*}
For $i=1,\dots,n$, let
\begin{eqnarray*}
\Gamma_5(\alpha_i)&=&\phi_1^{'} \Big(\frac{1-e^{-(\lambda x)^{k}}}{1-\bar{\alpha_i}e^{-(\lambda x)^{k}}}\Big)\frac{(-1)(1-e^{-(\lambda x)^k})e^{-(\lambda x)^k}}{(1-\bar{\alpha_i}e^{-(\lambda x)^k} )^2}.
\end{eqnarray*}
Upon partial derivative of $\Gamma_5(\alpha_i)$ with respect to  $\alpha_i$, we get 
{\fontsize{9}{12}\selectfont\begin{eqnarray*}
\frac{\partial \Gamma_5(\alpha_i)}{\partial \alpha_i}&=&\phi_1^{''}\Big(\frac{1-e^{-(\lambda x)^{k}}}{1-\bar{\alpha_i}e^{-(\lambda x)^{k}}}\Big)\Big(\frac{(-1)(1-e^{-(\lambda x)^k})e^{-(\lambda x)^k}}{(1-\bar{\alpha_i}e^{-(\lambda x)^k} )^2}\Big)^2+2\phi_1^{'}\Big(\frac{1-e^{-(\lambda x)^{k}}}{1-\bar{\alpha_i}e^{-(\lambda x)^{k}}}\Big)\frac{(1-e^{-(\lambda x)^{k}})e^{-2(\lambda x)^k}}{(1-\bar{\alpha_i}e^{-(\lambda x)^{k}})^3}\\
&=&\frac{(1-e^{-(\lambda x)^{k}})e^{-2(\lambda x)^k}}{(1-\bar{\alpha_i}e^{-(\lambda x)^{k}})^3}\Bigg[\phi_1^{''} \Big(\frac{1-e^{-(\lambda x)^{k}}}{1-\bar{\alpha_i}e^{-(\lambda x)^{k}}}\Big)\frac{1-e^{-(\lambda x)^{k}}}{1-\bar{\alpha_i}e^{-(\lambda x)^{k}}}+2\phi_1^{'} \Big(\frac{1-e^{-(\lambda x)^{k}}}{1-\bar{\alpha_i}e^{-(\lambda x)^{k}}}\Big)\Bigg]\geq 0,
\end{eqnarray*}}
since $
t\phi_1^{''}(t)+2\phi_1^{'}(t)\geq 0.$
Hence, 
$$ 
    (\alpha_i-\alpha_j)\Big(\frac{\partial \delta_5(\boldsymbol{\alpha})}{\partial \alpha_{i}}-\frac{\partial \delta_5(\boldsymbol{\alpha})}{\partial \alpha_{j}}\Big)
    =(\alpha_i-\alpha_j)\psi_{1}'\Big[\sum_{m=1}^{n}\phi_{1}\Big(\frac{1-e^{-(x\lambda)^k}}{1-\bar{\alpha_m}e^{-(x\lambda)^k}}\Big)\Big][\Gamma_5(\alpha_i)-\Gamma_5(\alpha_j)]\leq 0 .
$$
{and so, $\delta_5(\boldsymbol{\alpha})$ is decreasing and Schur-concave in $\alpha,$ from Lemma \ref{schur-critetia}. This completes the proof of the theorem.}
\end{proof}
Similarly, we can also derive conditions under which two series systems are comparable when the tilt parameter vectors are connected by weakly super-majorization order, as done in the following theorem.

\begin{theorem}\label{th7}
Let $X_i\sim EW(\alpha_i,\lambda,k) $ $(i=1,\ldots,n)$ and $Y_i\sim EW(\beta_i, \lambda, k)$ for $(i=1,\ldots,n)$ and the associated Archimedean survival copulas be with generators $\psi_1$ and $\psi_2$, respectively. Further, let $\phi_2\circ\psi_1$ be super-additive. Then, 
$$ \boldsymbol{\alpha}\succeq^{w}\boldsymbol{\beta} \Rightarrow X_{1:n}\succeq_{st}Y_{1:n}.$$
\end{theorem}
\begin{proof}
First, to establish the required result, we only need to prove that 
\begin{equation}
    \psi_{1}\Big[\sum_{m=1}^{n} \phi_{1}\Big(\frac{\beta_{m}e^{-(x\lambda)^{k}}}{1-\bar{\beta_m}e^{-(x\lambda)^{k}}}\Big)\Big]\geq \psi_{1}\Big[\sum_{m=1}^{n} \phi_{1}\Big(\frac{\alpha_{m}e^{-(x\lambda)^{k}}}{1-\bar{\alpha_m}e^{-(x\lambda)^{k}}}\Big)\Big].
\end{equation}
Let us define     
    \begin{eqnarray*}
        \delta_{51}(\boldsymbol{\alpha})&=&\psi_1\Big[\sum_{m=1}^n \phi_1 \Big(\frac{\alpha_{m}e^{-(\lambda x)^{k}}}{1-\bar{\alpha_m}e^{-(\lambda x)^{k}}}\Big)\Big].
    \end{eqnarray*} Upon differentiating $\delta_{51}(\boldsymbol{\alpha})$ with respect to $\alpha_i$, we get 
\begin{eqnarray*}
\frac{\partial \delta_{51}(\boldsymbol{\alpha})}{\partial \alpha_i}&=&\psi_1^{'}\Big[\sum_{i=1}^n \phi_1 \Big(\frac{\alpha_{m}e^{-(\lambda x)^{k}}}{1-\bar{\alpha_i}e^{-(\lambda x)^{k}}}\Big)\Big]\phi_1^{'} \Big(\frac{\alpha_{m}e^{-(\lambda x)^{k}}}{1-\bar{\alpha_i}e^{-(\lambda x)^{k}}}\Big)\frac{(1-e^{-(\lambda x)^k})e^{-(\lambda x)^k}}{(1-\bar{\alpha_i}e^{-(\lambda x)^k} )^2} \geq 0.
\end{eqnarray*}
Now, for $i=1,\dots,n$, let  
\begin{eqnarray*}
\Gamma_{51}(\alpha_i)&=&\phi_1^{'} \Big(\frac{\alpha_{m}e^{-(\lambda x)^{k}}}{1-\bar{\alpha_i}e^{-(\lambda x)^{k}}}\Big)\frac{(1-e^{-(\lambda x)^k})e^{-(\lambda x)^k}}{(1-\bar{\alpha_i}e^{-(\lambda x)^k} )^2}.
\end{eqnarray*}
Upon taking partial derivative of $\Gamma_{51}(\alpha_i)$ with respect to  $\alpha_i$, we get 
{\fontsize{9}{12}\selectfont\begin{eqnarray*}
\frac{\partial \Gamma_{51}(\alpha_i)}{\partial \alpha_i}&=&\phi_1^{''}\Big(\frac{\alpha_{m}e^{-(\lambda x)^{k}}}{1-\bar{\alpha_i}e^{-(\lambda x)^{k}}}\Big)\Big(\frac{(1-e^{-(\lambda x)^k})e^{-(\lambda x)^k}}{(1-\bar{\alpha_i}e^{-(\lambda x)^k} )^2}\Big)^2-2\phi_1^{'}\Big(\frac{\alpha_{m}e^{-(\lambda x)^{k}}}{1-\bar{\alpha_i}e^{-(\lambda x)^{k}}}\Big)\frac{(1-e^{-(\lambda x)^{k}})e^{-2(\lambda x)^k}}{(1-\bar{\alpha_i}e^{-(\lambda x)^{k}})^3}\\
&=&\frac{(1-e^{-(\lambda x)^{k}})e^{-2(\lambda x)^k}}{(1-\bar{\alpha_i}e^{-(\lambda x)^{k}})^3}\Bigg[\phi_1^{''} \Big(\frac{\alpha_{m}e^{-(\lambda x)^{k}}}{1-\bar{\alpha_i}e^{-(\lambda x)^{k}}}\Big)\frac{1-e^{-(\lambda x)^{k}}}{1-\bar{\alpha_i}e^{-(\lambda x)^{k}}}-2\phi_1^{'} \Big(\frac{\alpha_{m}e^{-(\lambda x)^{k}}}{1-\bar{\alpha_i}e^{-(\lambda x)^{k}}}\Big)\Bigg]\geq 0.
\end{eqnarray*}}

Hence, 
\begin{align}\label{th7_1}
    (\alpha_i-\alpha_j)\Big(\frac{\partial \delta_{51}(\boldsymbol{\alpha})}{\partial \alpha_{i}}-\frac{\partial \delta_{51}(\boldsymbol{\alpha})}{\partial \alpha_{j}}\Big)
    &=(\alpha_i-\alpha_j)\psi_{1}'\Big[\sum_{m=1}^{n}\phi_{1}\Big(\frac{\alpha_{m}e^{-(x\lambda)^k}}{1-\bar{\alpha_m}e^{-(x\lambda)^k}}\Big)\Big][\Gamma_5(\alpha_i)-\Gamma_5(\alpha_j)]\nonumber\\
    &\leq 0 
\end{align}
and so we have $\delta_{51}(\boldsymbol{\alpha})$ to be increasing and Schur-concave in $\alpha,$ from Lemma \ref{schur-critetia}. This completes the proof of the theorem.
\end{proof}    

For the purpose of illustrating Theorem \ref{th6}, present the following example. 
\begin{example}\label{ex3}
	Let $X_i \sim EW(\alpha_i,\lambda,k)$ and $Y_i \sim EW(\beta_i,\mu,k)$ for $i=1,2,3$. Set $k=5.67$, $\lambda= 5.37$  $(\alpha_1, \alpha_2, \alpha_3)=(0.4, 0.9, 0.1)$ and $(\beta_1, \beta_2, \beta_3)=(0.5, 0.8, 0.1)$. It is now easy to see that $(\alpha_1, \alpha_2, \alpha_3)\succeq_{w}(\beta_1, \beta_2, \beta_3)$.  
	Now, suppose we choose the copula $\phi(t)=e^{\frac{\theta}{t}}-e^{\theta}$, for $t\geq 0$, that satisfies $$ t\phi^{''}(t)+2\phi^{'}(t)\geq0,$$ where $t\geq0,$ and the parameters $\theta_1=2.4$ and $\theta_2=3.7$ which ensure the super-additivity of $\phi_{2} \circ \psi_1$. The distribution functions of $X_{3:3}$ and $Y_{3:3}$ are, respectively,
	
	\begin{equation*}
	F_{X_{3:3}} (x)=\frac{\theta_{1}}{\ln\left(e^{\theta_1\frac{1-\bar{\alpha_{1}}{\mathrm{e}}^{-(\lambda x)^{k}}}{1-{\mathrm{e}}^{-(\lambda x)^{k}}}}+e^{\theta_1\frac{1-\bar{\alpha_{2}}{\mathrm{e}}^{-(\lambda x)^{k}}}{1-{\mathrm{e}}^{-(\lambda x)^{k}}}}+e^{\theta_1\frac{1-\bar{\alpha_{3}}{\mathrm{e}}^{-(\lambda x)^{k}}}{1-{\mathrm{e}}^{-(\lambda x)^{k}}}}-2\mathrm{e}^{\theta_{1}}\right)}
	\end{equation*}
	and
	\begin{equation*}
	F_{Y_{3:3}} (x)=\frac{\theta_{2}}{\ln\left(e^{\theta_2\frac{1-\bar{\beta_{1}}{\mathrm{e}}^{-(\lambda x)^{k}}}{1-{\mathrm{e}}^{-(\lambda x)^{k}}}}+e^{\theta_2\frac{1-\bar{\beta_{2}}{\mathrm{e}}^{-(\lambda x)^{k}}}{1-{\mathrm{e}}^{-(\lambda x)^{k}}}}+e^{\theta_2\frac{1-\bar{\beta_{3}}{\mathrm{e}}^{-(\lambda x)^{k}}}{1-{\mathrm{e}}^{-(\lambda x)^{k}}}}-2\mathrm{e}^{\theta_{2}}\right)}.
	\end{equation*}
Then, $F_{X_{3:3}}(x) \leq F_{Y_{3:3}}(x) ,$ for all $x\geq 0,$ as established in Theorem \ref{th6}.
\end{example}
As in Counterexample \ref{cex2}, in the counterexample below, we show that if we violate the condition "$\phi_2\circ\psi_1$ is super-additive'' in Theorem \ref{th6}, then the distribution functions of $X_{n:n}$ and $Y_{n:n}$ cross each other.
\begin{counterexample} \label{cex3}
Let $X_i \sim EW(\alpha_i,\lambda,k)$ and $Y_i \sim EW(\beta_i,\mu,k),$ for $i=1,2,3$.  Set $k=3.16$, $\lambda= 12.5$, $(\alpha_1, \alpha_2, \alpha_3)=(0.82,0.85,0.95)$ and $(\beta_1, \beta_2, \beta_3)=(0.4,0.84,0.87)$. It is easy to see that $(\alpha_1, \alpha_2, \alpha_3)\succeq_{w}(\beta_1, \beta_2, \beta_3)$.  
Now, suppose we choose the copula $\phi(t)=e^{\frac{\theta}{t}}-e^{\theta}$, for $t\geq 0$, that satisfies $$ t\phi^{''}(t)+2\phi^{'}(t)\geq0,$$ where $t\geq0,$ and the parameters $\theta_1=22.6$ and $\theta_2=10.7$ which violate the condition of super-additivity of $\phi_{2} \circ \psi_1$. 
\begin{figure} \label{fig3}
	\begin{center}
		\includegraphics[height=2.8in]{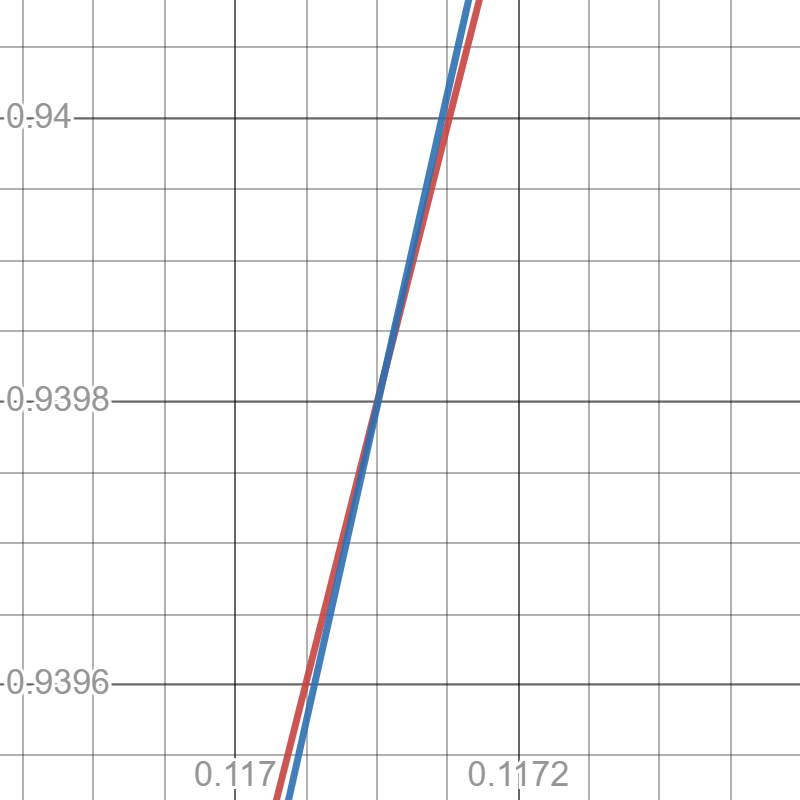}
		\caption{Plots of  ${F}_{X_{3:3}}(x)$ and ${F}_{Y_{3:3}}(x)$ as in Counterexample \ref{cex3}. Here, the red line corresponds to ${F}_{X_{3:3}}(x)$ and the blue line corresponds to ${F}_{Y_{3:3}}(x)$.}
	\end{center}
\end{figure}
Figure 4 plots $F_{X_{3:3}}(x)$ and $F_{Y_{3:3}}(x)$, and it is evident from it that when the condition of super-additivity of $\phi_{2} \circ \psi_1$ is violated in Theorem \ref{th6}, $F_{X_{3:3}}(x)$ is greater than $F_{Y_{3:3}}(x),$ for some $x\geq 0$.
\end{counterexample}
Next, we establish another result with regard to the comparison of $X_{1:n}$ and $Y_{1:n}$ in terms of the hazard rate order where the Archimedean survival copula is taken as independence copula with same generators.
\begin{theorem}
    Let $X_i\sim EW(\alpha,\lambda_{i},k)$ and $Y_i\sim EW(\alpha, \mu_{i}, k)$, for $i=1,\ldots,n$, and the associated Archimedean survival copulas are with generators $\psi_1=\psi_2=e^{-x}$, that is , $X_i$'s and $Y_i$'s are independent random variables. Also, let $0<\alpha\leq 1$ and ${k}\geq 1$. Then, 
$$ \boldsymbol{\lambda}\succeq^{m}\boldsymbol{\mu} \Rightarrow Y_{1:n}\succeq_{hr}X_{1:n}.$$
\end{theorem}
\begin{proof}
    Under the independence copula, the hazard rate function of $X_{1:n}$ is given by
    \begin{align}
        r_{X_{1:n}}(x)&=\sum_{i=1}^{n}\frac{k\lambda_{i}(\lambda_{i} x)^{k-1}}{1-\bar{\alpha}e^{-(\lambda_{i} x)^{k}}}.
    \end{align}
     By applying Lemma \ref{lemma2.7}, it is easy to observe that $r_{X_{1:n}}(x)$ is convex in $\bm{\lambda}.$ Now, upon using Proposition $C1$ of \cite{Marshall2011}, we observe that $r_{X_{1:n}}(x)$ is Schur-convex with respect to $\bm{\lambda}$, which proves the theorem.
\end{proof}

We then present two more results concerning the hazard rate order and the reversed hazard rate order between the smallest order statistics when the tilt parameters are connected in weakly super-majorization order and under independence copula with generator $\psi_1(x)=\psi_2(x)=e^{-x},~x>0.$ The proofs of two results can be completed by using Lemma $3.2$ of \cite{balakrishnan2018modified}.
\begin{theorem}
    Let $X_i\sim EW(\alpha_{i},\lambda,k)$ and $Y_i\sim EW(\beta_{i}, \lambda, k)$, for $i=1,\ldots,n$, and the associated Archimedean survival copulas be with generators $\psi_1=\psi_2=e^{-x}$, respectively. Then, 
$$ \boldsymbol{\alpha}\succeq^{w}\boldsymbol{\beta} \Rightarrow Y_{1:n}\succeq_{hr}X_{1:n}.$$
\end{theorem}
\begin{theorem}
    Let $X_i\sim EW(\alpha_{i},\lambda,k)$ and $Y_i\sim EW(\beta_{i}, \lambda, k)$, for $i=1,\ldots,n$, and the associated Archimedean survival copulas be with generators $\psi_1=\psi_2=e^{-x}$, respectively. Then, 
$$ \boldsymbol{\alpha}\succeq^{w}\boldsymbol{\beta} \Rightarrow Y_{n:n}\succeq_{rh}X_{n:n}.$$
\end{theorem}
We now derive some conditions on model parameters, for comparing the extremes with respect to the dispersive, star and Lorenz orders, when the variables are dependent and follow extended Weibull distributions structured with Archimedean copula having the same generator.

 Before stating our main resluts, we present the following two lemmas which will be used to prove the main results. 

\begin{lemma}(\cite{saunders1978})\label{lem_star}
    Let $X_{a}$ be a random variable with distribution function $F_{a}$, for each $a\in(0,\infty),$ such that
    \begin{itemize}
        \item[(i)] $F_{a}$ is supported on some interval $(x^{(a)}_{-},x^{(a)}_{+})\subset(0,\infty)$ and has density $f_{a}$ which does not vanish in any subinterval of $(x^{(a)}_{-},x^{(a)}_{+});$
        \item[(ii)] The derivative of $F_{a}$ with respect to $a$ exists and is denoted by $F^{'}_{a}.$ 
    \end{itemize}
Then, $X_{a}\geq_{*}X_{a^{*}},$ for $a,~a^{*}\in (0,\infty)$ and $a>a^{*},$ iff $F^{'}_{a}(x)/xf_{a}(x)$ is decreasing in $x.$    
\end{lemma}

We now establish some sufficient conditions for the comparison of two extremes in the sense of star order, with the first result being for parallel systems and the second one being for series systems. 
\begin{theorem}\label{Maj-star1}
Let $X_i \sim EW (\alpha,\lambda_1,k)$  $(i=1,\ldots,p)$ and $X_j \sim EW (\alpha,\lambda_2,k)$  $(j=p+1,\ldots,n)$, and $Y_i \sim EW (\alpha,\mu_1,k)$  $(i=1,\ldots,p)$ and $Y_j \sim EW (\alpha,\mu_2,k )$  $(j=p+1,\ldots,n)$ be variables with a common Archimedean copula having generator $\psi$.  Then, if 
\begin{equation*}
(\alpha+\bar{\alpha}t)\ln\left(\frac{t}{\alpha+\bar{\alpha}t}\right)\left[\alpha+2\bar{\alpha}t-\frac{(\alpha+\bar{\alpha}
t)t\phi''(1-t)}{\phi'(1-t)}\right]
\end{equation*}
is increasing with respect to $t \in [0,1]$ and $0<k\leq 1$, we have
\begin{eqnarray*}
(\lambda_1-\lambda_2)(\mu_1-\mu_2) \ge 0\,\,\,\,\,and\,\,\,\,\,\,
\frac{\lambda_{2:2}}{\lambda_{1:2}} \ge \frac{\mu_{2;2}}{\mu_{1:2}}
\Longrightarrow
Y_{n:n} \le_{*} X_{n:n},
\end{eqnarray*}
where $p+q=n$.
\end{theorem}

 \begin{proof} Assume that $(\lambda_1-\lambda_2)(\mu_1-\mu_2) \ge 0$. Now, without
loss of generality, let us assume that $\lambda_1\le\lambda_2$ and
$\mu_1\le\mu_2$. The distribution
functions of $X_{n:n}$ and $Y_{n:n}$ are 
\begin{eqnarray*}
F_{X_{n:n}}(x)=\psi\left[p\phi\left(\frac{1-e^{- (\lambda_1 x)^k}}{1-\bar{\alpha}e^{- (\lambda_1 x)^k}}
\right)+q\phi\left(\frac{1-e^{-(\lambda_2 x)^k}}{1-\bar{\alpha} e^{-(\lambda_2 x)^k}}\right)\right], \qquad x\in(0,\infty)
\end{eqnarray*}
and
\begin{eqnarray*}
F_{Y_{n:n}}(x)=\psi\left[p\phi\left(\frac{1-e^{- (\mu_1 x)^k}}{1-\bar{\alpha}e^{- (\mu_1 x)^k}}
\right)+q\phi\left(\frac{1-e^{-(\mu_2 x)^k}}{1-\bar{\alpha} e^{-(\mu_2 x)^k}}\right)\right], \qquad x\in(0,\infty),
\end{eqnarray*}
where $q=n-p$. In this case, the proof can be completed by considering the following two cases.\\
{\bf {Case (i)}:} \,\,$\lambda_1+\lambda_2=\mu_1+\mu_2.$ For convenience, let us assume that $\lambda_1+\lambda_2=\mu_1+\mu_2=1$. 
Set $\lambda_2=\lambda$, $\mu_2=\mu$,
$\lambda_1=1-\lambda$ and
$\mu_1=1-\mu$.  Then, under this setting, the distribution
functions of $X_{n:n}$ and $Y_{n:n}$ are 
\begin{eqnarray*}
F_{\lambda}(x)=\psi\left[p\phi\left(\frac{1-e^{-(1-\lambda)^k x^k}}{1-\bar{\alpha}e^{-(1-\lambda)^k x^k}}
\right)+q\phi\left(\frac{1-e^{-\lambda^k x^k}}{1-\bar{\alpha} e^{-\lambda^k x^k}}\right)\right], \qquad x\in(0,\infty),
\end{eqnarray*}
\begin{eqnarray*}
F_{\mu}(x)=\psi\left[p\phi\left(\frac{1-e^{-(1-\mu)^k x^k}}{1-\bar{\alpha}e^{-(1-\mu)^k x^k}}
\right)+q\phi\left(\frac{1-e^{-\mu^k x^k}}{1-\bar{\alpha} e^{-\mu^k x^k}}\right)\right], \qquad x\in(0,\infty),
\end{eqnarray*}
 respectively. Now, to obtain the required
result, it is sufficient to show that $\frac{F'_\lambda(x)}{x f_\lambda (x)}$ is decreasing  in $x \in(0,\infty),$ for $\lambda \in(1/2,1]$ by Lemma \ref{lem_star}.
The derivative of $F_{\lambda}$ with respect to $\lambda$ is

{\fontsize{9}{12}\selectfont\begin{eqnarray*}
F'_\lambda(x)&=&\psi'\left[p\phi\left(\frac{1-e^{-(1-\lambda)^{k} x^k}}{1-\bar{\alpha}e^{-(1-\lambda)^{k} x^k}}
\right)+q\phi\left(\frac{1-e^{-\lambda^k x^k}}{1-\bar{\alpha} e^{-\lambda^k x^k}}\right)\right]\\
&&\times \left[\frac{-p\alpha x^k e^{-(1-\lambda)^k x^k}}{(1-\bar{\alpha} e^{-(1-\lambda)^k x^k})^2}
\phi'\left(\frac{1-e^{-(1-\lambda)^k x^k}}{1-\bar{\alpha}e^{-(1-\lambda)^k x^k}} \right)k(1-\lambda)^{k-1}+q\frac{\alpha xe^{-\lambda^k x^k}}{(1-\bar{\alpha}e^{-\lambda^k x^k})^2}\phi'\left(\frac{1-e^{-\lambda^k x^k}}{1-\bar{\alpha}e^{-\lambda^k x^k}}\right)k\lambda^{k-1}\right].
\end{eqnarray*}}
On the other hand, the density function corresponding to
$F_{\lambda}$ has the form
{\fontsize{9}{12}\selectfont\begin{eqnarray*}
f_\lambda(x)&=&\psi'\left[p\phi\left(\frac{1-e^{-(1-\lambda)^k x^k}}{1-\bar{\alpha}e^{-(1-\lambda)^k x^k}}
\right)+q\phi\left(\frac{1-e^{-\lambda^k x^k}}{1-\bar{\alpha} e^{-\lambda^k x^k}}\right)\right]\\
&&\times \left[\frac{p\alpha (1-\lambda)^k e^{-(1-\lambda)^k x^k}}{(1-\bar{\alpha} e^{-(1-\lambda)^k x^k})^2}
\phi'\left(\frac{1-e^{-(1-\lambda)^k x^k}}{1-\bar{\alpha}e^{-(1-\lambda)^k x^k}} \right)+\frac{q\alpha \lambda^k e^{-\lambda^k x^k}}{(1-\bar{\alpha}e^{-\lambda^k x^k})^2}\phi'\left(\frac{1-e^{-\lambda^k x^k}}{1-\bar{\alpha}e^{-\lambda^k x^k}}\right)\right]kx^{k-1}.
\end{eqnarray*}}
So, we have 
{\fontsize{9}{12}\selectfont\begin{eqnarray*}
\frac{F'_\alpha(x)}{xf_\alpha(x)}&=&\frac{\frac{-p\alpha x^k e^{-(1-\lambda)^k x^k}}{(1-\bar{\alpha} e^{-(1-\lambda)^k x^k})^2}
\phi'\left(\frac{1-e^{-(1-\lambda)^k x^k}}{1-\bar{\alpha}e^{-(1-\lambda)^k x^k}} \right)k(1-\lambda)^{k-1}+q\frac{\alpha xe^{-\lambda^k x^k}}{(1-\bar{\alpha}e^{-\lambda^k x^k})^2}\phi'\left(\frac{1-e^{-\lambda^k x^k}}{1-\bar{\alpha}e^{-\lambda^k x^k}}\right)k\lambda^{k-1}}{\frac{p\alpha (1-\lambda)^k e^{-(1-\lambda)^k x^k}}{(1-\bar{\alpha} e^{-(1-\lambda)^k x^k})^2}
\phi'\left(\frac{1-e^{-(1-\lambda)^k x^k}}{1-\bar{\alpha}e^{-(1-\lambda)^k x^k}} \right)kx^k+\frac{q\alpha \lambda^k e^{-\lambda^k x^k}}{(1-\bar{\alpha}e^{-\lambda^k x^k})^2}\phi'\left(\frac{1-e^{-\lambda^k x^k}}{1-\bar{\alpha}e^{-\lambda^k x^k}}\right)kx^{k}}\\
&=&\left(\lambda+\frac{p\alpha x^k\frac{e^{-(1-\lambda)^k x^k}}{(1-\bar{\alpha} e^{-(1-\lambda)^k x^k})^2}\phi'\left(\frac{1-e^{(1-\lambda)^k x^k}}{1-\bar{\alpha}e^{(1-\lambda)^k x^k}}\right)k(1-\lambda)^{k-1}}{\frac{-p\alpha x^k e^{-(1-\lambda)^k x^k}}{(1-\bar{\alpha} e^{-(1-\lambda)^k x^k})^2}
\phi'\left(\frac{1-e^{-(1-\lambda)^k x^k}}{1-\bar{\alpha}e^{-(1-\lambda)^k x^k}} \right)k(1-\lambda)^{k-1}+q\frac{\alpha xe^{-\lambda^k x^k}}{(1-\bar{\alpha}e^{-\lambda^k x^k})^2}\phi'\left(\frac{1-e^{-\lambda^k x^k}}{1-\bar{\alpha}e^{-\lambda^k x^k}}\right)k\lambda^{k-1}}\right)^{-1}\\
&=&\left(\lambda+\left(\frac{\frac{qe^{-\lambda^k x^k}}{(1-\bar{\alpha}e^{-\lambda^k x^k})^2}\phi'\left(\frac{1-e^{-\lambda^k x^k}}{1-\bar{\alpha}e^{-\lambda^k x^k}}\right)\lambda^{k-1}}{\frac{p e^{-(1-\lambda)^k x^k}}{(1-\bar{\alpha} e^{-(1-\lambda)^k x^k})^2}\phi'\left(\frac{1-e^{(1-\lambda)^k x^k}}{1-\bar{\alpha}e^{(1-\lambda)^k x^k}}\right)(1-\lambda)^{k-1}}-1 \right)^{-1}\right)^{-1}.
\end{eqnarray*}}
Thus, it suffices to show that, for $\lambda\in (1/2,1]$,
\begin{equation*}
\Delta(x)=\frac{\frac{e^{-\lambda^k x^k}}{(1-\bar{\alpha}e^{-\lambda^k x^k})^2}\phi'\left(\frac{1-e^{-\lambda x^k}}{1-\bar{\alpha}e^{-\lambda^k x^k}}\right)\lambda^{k-1}}{\frac{e^{-(1-\lambda)^k x^k}}{(1-\bar{\alpha} e^{-(1-\lambda)^k x^k})^2}\phi'\left(\frac{1-e^{(1-\lambda)^k x^k}}{1-\bar{\alpha}e^{(1-\lambda)^k x^k}}\right)(1-\lambda)^{k-1}}
\end{equation*}
is decreasing in $x\in(0,\infty)$.  Now, let us set $t_1=\frac{\alpha  e^{-\lambda^k x^k}}{1-\bar{\alpha}e^{-\lambda^k x^k}}$ and $t_2=\frac{\alpha  e^{-(1-\lambda)^k x^k}}{1-\bar{\alpha}e^{-(1-\lambda)^k x^k}}$. 
From  the fact that  $\lambda \in (1/2, 1]$, we have  $t_1< t_2$ for all $x\in (0,\infty)$, and so 
{\begin{eqnarray*}
\Delta(x)&=&\frac{t_1^2\phi'(1-t_1)e^{\lambda^k x^k}}{t_2^2\phi'(1-t_2)e^{(1-\lambda)^k x^k}}\\
&=&\frac{t_1(\alpha+\bar{\alpha}t_1)\phi'(1-t_1)}{t_2(\alpha+\bar{\alpha}t_2)\phi'(1-t_2)}
\end{eqnarray*}}
from which we get the derivative of $\Delta(x)$ with respect to $x$ to be
\begin{eqnarray*}
\Delta'(x)&\stackrel{sgn}{=}&\left[t_1'(\alpha+\bar{\alpha}t_1)\phi'(1-t_1)-t_1(\alpha+\bar{\alpha}t_1)t_1'\phi''(1-t_1)+\bar{\alpha}t_1t_1'\phi'(1-t_1)\right]\times t_2(\alpha+\bar{\alpha}t_2)\phi'(1-t_2)\\
&&-\left[t_2'(\alpha+\bar{\alpha}t_2)\phi'(1-t_2)-t_2(\alpha+\bar{\alpha}t_2)t_2'\phi''(1-t_2)+\bar{\alpha}t_2t_2'\phi'(1-t_2)\right]\times t_1(\alpha+\bar{\alpha}t_1)\phi'(1-t_1)\\
&\stackrel{sgn}{=}&\frac{t_1'}{t_1}\left[\alpha+2\bar{\alpha}t_1-\frac{(\alpha+\bar{\alpha}
t_1)t_1\phi''(1-t_1)}{\phi'(1-t_1)}\right]-\frac{t_2'}{t_2}\left[\alpha+2\bar{\alpha}t_2-\frac{(\alpha+\bar{\alpha}
t_2)t_2\phi''(1-t_2)}{\phi'(1-t_2)}\right].
\end{eqnarray*}
It is easy to show that the derivatives of $t_1$ and $t_2$ with respect to $x$ are
\begin{eqnarray*}
t_1'=\frac{-\lambda^k t_1 k x^{k-1}}{1-\bar{\alpha}e^{-\lambda^k x^k}}=\frac{-\lambda^k k x^{k-1}}{\alpha}(\alpha+\bar{\alpha}t_1)t_1=\frac{k(\alpha+\bar{\alpha}t_1)t_1}{\alpha x}\ln\left(\frac{t_1}{\alpha+\bar{\alpha}t_1}\right),
\end{eqnarray*}
\begin{eqnarray*}
t_2'=\frac{-(1-\lambda)^k k x^{k-1} t_2}{1-\bar{\alpha}e^{(1-\lambda)^k x^k}}=\frac{-(1-\lambda)^k k x^{k-1}}{\alpha}(\alpha+\bar{\alpha}t_2)t_2=\frac{k(\alpha+\bar{\alpha}t_2)t_2}{\alpha x}\ln\left(\frac{t_2}{\alpha+\bar{\alpha}t_2}\right).
\end{eqnarray*}
Hence, we get
\begin{eqnarray*}
\Delta'(x)&\stackrel{sgn}{=}& k(\alpha+\bar{\alpha}t_1)\ln\left(\frac{t_1}{\alpha+\bar{\alpha}t_1}\right)\left[\alpha+2\bar{\alpha}t_1-\frac{(\alpha+\bar{\alpha}
t_1)t_1\phi''(1-t_1)}{\phi'(1-t_1)}\right]\\
&&-k(\alpha+\bar{\alpha}t_2)\ln\left(\frac{t_2}{\alpha+\bar{\alpha}t_2}\right)\left[\alpha+2\bar{\alpha}t_2-\frac{(\alpha+\bar{\alpha}
t_2)t_2\phi''(1-t_2)}{\phi'(1-t_2)}\right].
\end{eqnarray*}
As $t_1<t_2$,  $\Delta'<0$ if and only if 
\[(\alpha+\bar{\alpha}t)\ln\left(\frac{t}{\alpha+\bar{\alpha}t}\right)\left[\alpha+2\bar{\alpha}t-\frac{(\alpha+\bar{\alpha}
t)t\phi''(1-t)}{\phi'(1-t)}\right]\]
is increasing in $t\in[0,1]$. 

{\bf {Case (ii)}.}\,\,$\lambda_1+\lambda_2\neq \mu_1+\mu_2.$  In this case, we can note that $\lambda_1+\lambda_2=k(\mu_1+\mu_2)$,  where $k$ is
a scalar. We then have $(k\mu_1,k\mu_2)\stackrel{m}{\preceq}(\lambda_1 , \lambda_2)$. 
 Let $W_{1:n}$ be the lifetime of a series system
having $n$ dependent extended exponentially  distributed components whose lifetimes have an
Archimedean copula with generator $\psi$, where $W_i \sim  EW(\alpha, k,\mu_1)$ ($i=1,\ldots, p$) and $W_{j} \sim EW(\alpha,k,\mu_2)$
 ($j=p+1,\ldots,n$). From the result in  Case (i), we then have $  W_{n:n} \le_{*} X_{n:n}$. But,
since  star order is scale invariant, it then follows that $ Y_{n:n} \le_{*} X_{n:n}$. 
\end{proof}
\begin{theorem}\label{M-star2}
Let $X_i \sim EW (\alpha,\lambda_1,k)$ $(i=1,\ldots,p)$  and $X_j \sim EW (\alpha,\lambda_2,k)$ $(j=p+1,\ldots,n)$, and $Y_i \sim EW (\alpha,\mu_1,k)$  $(i=1,\ldots,p)$ and $Y_j \sim EW (\alpha,\mu_2,k)$ $(j=p+1,\ldots,n)$ be variables with a common Archimedean survival
copula having generator $\psi$. Then, if
\[(\alpha+\bar{\alpha}t)\ln\left(\frac{t}{\alpha+\bar{\alpha}t}\right)\left((\alpha +2\bar{\alpha}t)+(\alpha t+\bar{\alpha}t^2)\frac{\phi''(t)}{\phi'(t)}\right)\]
is decreasing with respect to $t \in [0,1]$ and $0\le k\leq 1$, we have
\begin{equation}
(\lambda_1-\lambda_2)(\mu_1-\mu_2) \ge 0\,\,\,\,\,and\,\,\,\,\,\,
\frac{\lambda_{2:2}}{\lambda_{1:2}} \ge \frac{\mu_{2;2}}{\mu_{1:2}}
\Longrightarrow
Y_{1:n} \le_{*} X_{1:n},
\end{equation}
where $p+q=n$.
\end{theorem}
\begin{proof} Assume that $(\lambda_1-\lambda_2)(\mu_1-\mu_2) \ge 0$. Now, without
loss of generality, let us assume that $\lambda_1\le\lambda_2$ and
$\mu_1\le\mu_2$. The distribution
functions of $X_{1:n}$ and $Y_{1:n}$ are 
\begin{equation}
F_{X_{1:n}}(x)=1-\psi\left[p\phi\left(\frac{\alpha e^{-(\lambda_1 x)^k}}{1-\bar{\alpha} e^{-(\lambda_1 x)^k}}\right)+q\phi\left(\frac{\alpha e^{- (\lambda_2 x)^k}}{1-\bar{\alpha}e^{- (\lambda_2 x)^k}} \right)\right],\quad x\in (0,\infty),
\end{equation}
\begin{equation}
F_{X_{1:n}}(x)=1-\psi\left[p\phi\left(\frac{\alpha e^{-(\mu_1 x)^k}}{1-\bar{\alpha} e^{-(\mu_1 x)^k}}\right)+q\phi\left(\frac{\alpha e^{- (\mu_2 x)^k}}{1-\bar{\alpha}e^{-( \mu_2 x)^k}} \right)\right],\quad x\in (0,\infty),
\end{equation}
where $q=n-p$. In this case, the proof can be completed by  considering the following two cases.\\
{\bf {Case (i)}:} \,\,$\beta_1+\beta_2=\mu_1+\mu_2.$ For convenience, let us assume that $\beta_1+\beta_2=\mu_1+\mu_2=1$. 
Set $\lambda_1=\lambda$, $\mu_1=\mu$,
$\lambda_2=1-\lambda$ and
$\mu_1=1-\mu$. Then, under this setting, the distribution
functions of $X_{1:n}$ and $Y_{1:n}$ are 
\begin{equation}
F_{\lambda}(x)=1-\psi\left[p\phi\left(\frac{\alpha e^{-\lambda^k x^k}}{1-\bar{\alpha} e^{-\lambda^k x^k}}\right)+q\phi\left(\frac{\alpha e^{-(1-\lambda)^k x^k}}{1-\bar{\alpha}e^{-(1-\lambda)^k x^k}} \right)\right],\quad x\in (0,\infty),
\end{equation}
and
\begin{eqnarray*}
F_{\mu}(x)=1-\psi\left[p\phi\left(\frac{\alpha e^{-\mu^k x^k}}{1-\bar{\alpha} e^{-\mu^k x^k}}\right)+q\phi\left(\frac{\alpha e^{-(1-\mu)^k x^k}}{1-\bar{\alpha}e^{-(1-\mu)^k x^k}} \right)\right],\quad x\in (0,\infty),
\end{eqnarray*}
respectively. Now, to obtain the required
result, it is sufficient to show that $\frac{F'_\lambda(x)}{x f_\lambda (x)}$ decreasing  in $x \in(0,\infty)$ for $\alpha \in[0,1/2)$.
The derivative of $F_{\lambda}$ with respect to $\lambda$ is
{\fontsize{9}{12}\selectfont\begin{eqnarray*}
F'_\lambda(x)&=&-\psi'\left[p\phi\left(\frac{\alpha e^{-\lambda^k x^k}}{1-\bar{\alpha}e^{-\lambda^k x^k}}
\right)+q\phi\left(\frac{\alpha e^{-(1-\lambda)^k x^k}}{1-\bar{\alpha} e^{-(1-\lambda)^k x^k}}\right)\right]\\
&&\times \left[p\frac{-\alpha x^k e^{-\lambda^k x^k}}{(1-\bar{\alpha} e^{-\lambda^k x^k})^2}
\phi'\left(\frac{\alpha e^{-\lambda^k x^k}}{1-\bar{\alpha}e^{-\lambda^k x^k}} \right)k\lambda^{k-1}+q\frac{\alpha xe^{-(1-\lambda)^k x^k}}{(1-\bar{\alpha}e^{-(1-\lambda)^k x^k})^2}\phi'\left(\frac{\alpha e^{-(1-\lambda)^k x^k}}{1-\bar{\alpha}e^{-(1-\lambda)^k x^k}}\right)k(1-\lambda)^{k-1}\right].
\end{eqnarray*}}
On the other hand, the density function corresponding to
$F_{\lambda}$ has the form
{\fontsize{9.5}{12}\selectfont\begin{eqnarray*}
f_\lambda(x)&=&-\psi'\left[p\phi\left(\frac{\alpha e^{-\lambda^k  x^k}}{1-\bar{\alpha}e^{-\lambda^k x^k}}
\right)+q\phi\left(\frac{\alpha e^{-(1-\lambda)^k x^k}}{1-\bar{\alpha} e^{-(1-\lambda)^k x^k}}\right)\right]\\
&&\times \left[p\frac{-\alpha \lambda^k e^{-\lambda^k x^k}}{(1-\bar{\alpha} e^{-\lambda^k x^k})^2}
\phi'\left(\frac{\alpha e^{-\lambda^k x^k}}{1-\bar{\alpha}e^{-\lambda^k x^k}} \right)-q\frac{(1-\lambda)^k \alpha e^{-(1-\lambda)^k x^k}}{(1-\bar{\alpha}e^{-(1-\lambda)^k x^k})^2}\phi'\left(\frac{\alpha e^{-(1-\lambda)^k x^k}}{1-\bar{\alpha}e^{-(1-\lambda)^k x^k}}\right)\right]kx^{k-1}.
\end{eqnarray*}}
So, we have 
{\fontsize{9}{12}\selectfont\begin{eqnarray*}
\frac{F'_\lambda(x)}{xf_\lambda(x)}&=& \frac{p\frac{-\alpha x^k e^{-\lambda^k x^k}}{(1-\bar{\alpha} e^{-\lambda^k x^k})^2}
\phi'\left(\frac{\alpha e^{-\lambda^k x^k}}{1-\bar{\alpha}e^{-\lambda^k x^k}} \right)\lambda^{k-1}+q\frac{\alpha x^k e^{-(1-\lambda)^k x^k}}{(1-\bar{\alpha}e^{-(1-\lambda)^k x^k})^2}\phi'\left(\frac{\alpha e^{-(1-\lambda)^k x^k}}{1-\bar{\alpha}e^{-(1-\lambda)^k x^k}}\right)(1-\lambda)^{k-1}}{p\frac{-\alpha x^k\lambda^k e^{-\lambda^k x^k}}{(1-\bar{\alpha} e^{-\lambda^k x^k})^2}
\phi'\left(\frac{\alpha e^{-\lambda^k x^k}}{1-\bar{\alpha}e^{-\lambda^k x^k}} \right)-q\frac{(1-\lambda)^k x^k \alpha e^{-(1-\lambda)^k x^k}}{(1-\bar{\alpha}e^{-(1-\lambda)^k x^k})^2}\phi'\left(\frac{\alpha e^{-(1-\lambda)^k x^k}}{1-\bar{\alpha}e^{-(1-\lambda)^k x^k}}\right)}\\
&=&\left(\lambda+\frac{\frac{-q\alpha x^k e^{-(1-\lambda)^k x^k}}{(1-\bar{\alpha} e^{-(1-\lambda)^k x^k})^2}\phi'\left(\frac{\alpha e^{-(1-\lambda)^k x^k}}{1-\bar{\alpha}e^{-(1-\lambda)^k x^k}}\right)(1-\lambda)^{k-1}}{-p\frac{\alpha x^k e^{-\lambda^k x^k}}{(1-\bar{\alpha} e^{-\lambda^k x^k})^2}
\phi'\left(\frac{\alpha e^{-\lambda^k x^k}}{1-\bar{\alpha}e^{-\lambda^k x^k}} \right)\lambda^{k-1}+q\frac{\alpha x^k e^{-(1-\lambda)^k x^k}}{(1-\bar{\alpha}e^{-(1-\lambda)^k x^k})^2}\phi'\left(\frac{\alpha e^{-(1-\lambda)^k x^k}}{1-\bar{\alpha}e^{-(1-\lambda) x}}\right)(1-\lambda)^{k-1}} \right)^{-1}\\
&=&\left(\lambda+\left[\frac{p\frac{e^{-\lambda^k x^k}}{(1-\bar{\alpha}e^{-\lambda^k x^k})^2}\phi'\left(\frac{\alpha e^{-\lambda^k x^k}}{1-\bar{\alpha}e^{-\lambda^k x^k}}\right)\lambda^{k-1}}{q\frac{e^{-(1-\lambda)^k x^k}}{(1-\bar{\alpha} e^{-(1-\lambda)^k x^k})^2}\phi'\left(\frac{\alpha e^{-(1-\lambda)^k x^k}}{1-\bar{\alpha}e^{-(1-\lambda)^k x^k}}\right)(1-\lambda)^{k-1}}-1 \right]^{-1}\right)^{-1}.
\end{eqnarray*}}
We can then conclude that 
 $\frac{F'_\lambda (x)}{xf_\lambda (x)}$ is decreasing if 
\begin{equation*}
\Omega(x)=\frac{e^{\lambda^k x^k}\left(\frac{\alpha e^{-\lambda^k x^k}}{1-\bar{\alpha}e^{-\lambda^k x^k}}\right)^2\phi'\left(\frac{\alpha e^{-\lambda^k x^k}}{1-\bar{\alpha}e^{-\lambda^k x^k}}\right)}{e^{(1-\lambda)^k x^k}\left(\frac{\alpha e^{-(1-\lambda)^k x^k}}{1-\bar{\alpha} e^{-(1-\lambda)^k x^k}}\right)^2\phi'\left(\frac{\alpha e^{(1-\lambda)^k x^k}}{1-\bar{\alpha}e^{(1-\lambda)^k x^k}}\right)}
\end{equation*}
 is decreasing for $x\in (0,\infty)$. Let $t_1=\frac{\alpha e^{-\lambda^k x^k}}{1-\bar{\alpha}e^{-\lambda^k x^k}}$ and $t_2=\frac{\alpha e^{-(1-\lambda)^k x^k}}{1-\bar{\alpha}e^{-(1-\lambda)^k x^k}}$. 
 Then, $e^{\lambda^k x^k}=\frac{\alpha+\bar{\alpha}t_1}{t_1}$ and $e^{(1-\lambda)^k x^k}=\frac{\alpha+\bar{\alpha}t_2}{t_2}$, and so 
\[\Omega(x)=\frac{t_1(\alpha+\bar{\alpha}t_1)\phi'\left(t_1\right)}{t_2(\alpha+\bar{\alpha}t_2)\phi'\left(t_2\right)},\] 
whose derivative with respect to $x$ is 
\begin{eqnarray*}
\Omega'(x)&=&\left(\frac{t_1(\alpha+\bar{\alpha}t_1)\phi'\left(t_1\right)}{t_2(\alpha+\bar{\alpha}t_2)\phi'\left(t_2\right)}\right)'\\
&\stackrel{sgn}{=}&\left((\alpha t_1'+2\bar{\alpha}t_1t_1')\phi'(t_1)+(\alpha t_1+\bar{\alpha}t_1^2)t_1'\phi''(t_1)\right)\times (\alpha t_2+\bar{\alpha}t_2^2)\phi'(t_2)\\
&&-\left((\alpha t_2'+2\bar{\alpha}t_2t_2')\phi'(t_2)+(\alpha t_2+\bar{\alpha}t_2^2)t_2'\phi''(t_2)\right)\times (\alpha t_1+\bar{\alpha}t_1^2)\phi'(t_1)\\
&\stackrel{sgn}{=}&\frac{t_1'}{t_1}\left((\alpha +2\bar{\alpha}t_1)+(\alpha t_1+\bar{\alpha}t_1^2)\frac{\phi''(t_1)}{\phi'(t_1)}\right)\\
&&-\frac{t_2'}{t_2}\left((\alpha +2\bar{\alpha}t_2)+(\alpha t_2+\bar{\alpha}t_2^2)\frac{\phi''(t_2)}{\phi'(t_2)}\right).
\end{eqnarray*}
It is easy to show that 
\begin{equation*}
t_1'=\frac{(\alpha+\bar{\alpha}t_1)t_1}{\alpha x}\ln\left(\frac{t_1}{\alpha+\bar{\alpha}t_1}\right),\qquad t_2'=\frac{(\alpha+\bar{\alpha}t_2)t_2}{\alpha x}\ln\left(\frac{t_2}{\alpha+\bar{\alpha}t_2}\right).
\end{equation*}
Hence, we have
\begin{eqnarray*}
\Omega'(x)&\stackrel{sgn}{=}&(\alpha+\bar{\alpha}t_1)\ln\left(\frac{t_1}{\alpha+\bar{\alpha}t_1}\right)\left((\alpha +2\bar{\alpha}t_1)+(\alpha t_1+\bar{\alpha}t_1^2)\frac{\phi''(t_1)}{\phi'(t_1)}\right)\\
&-&(\alpha+\bar{\alpha}t_2)\ln\left(\frac{t_2}{\alpha+\bar{\alpha}t_2}\right)\left((\alpha +2\bar{\alpha}t_2)+(\alpha t_2+\bar{\alpha}t_2^2)\frac{\phi''(t_2)}{\phi'(t_2)}\right).
\end{eqnarray*}
Now, as $t_2<t_1$, $\Omega(x)$ is decreasing if 
\[(\alpha+\bar{\alpha}t)\ln\left(\frac{t}{\alpha+\bar{\alpha}t}\right)\left((\alpha +2\bar{\alpha}t)+(\alpha t+\bar{\alpha}t^2)\frac{\phi''(t)}{\phi'(t)}\right)\]
is decreasing in $t\in[0,1]$, as required.

 {\bf {Case (ii)}.}\,\,$\lambda_1+\lambda_2\neq \mu_1+\mu_2.$  In this case, we can note that $\lambda_1+\lambda_2=k(\mu_1+\mu_2)$, where $k$ is
a scalar. We then have $(k\mu_1,k\mu_2)\stackrel{m}{\preceq}(\lambda_1 , \lambda_2)$. 
 Let $W_{1:n}$ be the lifetime of a series system
having $n$ dependent extended exponentially distributed components whose lifetimes have an
Archimedean copula with generator $\psi$, where $W_i \sim  EW(\alpha, k\mu_1)$ ($i=1,\ldots, p$) and $W_{j} \sim EW(\alpha,k\mu_2)$
 ($j=p+1,\ldots,n$). From the result in Case (i), we then have $  W_{1:n} \le_{*} X_{1:n}$. But,
since  star order is scale invariant, it  then follows that $ Y_{1:n} \le_{*} X_{1:n}$. 
\end{proof}
It is important to mention that $X\leq_{*}Y$ implies $X\leq_{Loenz}Y.$ Therefore, from Theorems \ref{Maj-star1} and \ref{M-star2}, we readily obtain the following two corollaries.

\begin{corollary}\label{Cor-M-star2}
Let $X_i \sim EW (\alpha,\lambda_1,k)$   $(i=1,\ldots,p)$ and $X_j \sim EW (\alpha,\lambda_2,k)$   $(j=p+1,\ldots,n)$, and $Y_i \sim EW (\alpha,\mu_1,k)$   $(i=1,\ldots,p)$ and $Y_j \sim EW (\alpha,\mu_2,k)$   $(j=p+1,\ldots,n)$ be with a common Archimedean survival
copula having generator $\psi$.  Then, if 
\[(\alpha+\bar{\alpha}t)\ln\left(\frac{t}{\alpha+\bar{\alpha}t}\right)\left((\alpha +2\bar{\alpha}t)+(\alpha t+\bar{\alpha}t^2)\frac{\phi''(t)}{\phi'(t)}\right)\]
is decreasing with respect to $t \in [0,1]$ and $0\le k\leq 1$, we have
\begin{equation}
(\lambda_1-\lambda_2)(\mu_1-\mu_2) \ge 0\,\,\,\,\,and\,\,\,\,\,\,
\frac{\lambda_{2:2}}{\lambda_{1:2}} \ge \frac{\mu_{2;2}}{\mu_{1:2}}
\Longrightarrow
Y_{1:n} \le_{Lorenz} X_{1:n},
\end{equation}
where $p+q=n$.
\end{corollary}
\begin{corollary}\label{Cor-Maj-star1}
Let $X_i \sim EW (\alpha,\lambda_1,k)$   $(i=1,\ldots,p)$ and $X_j \sim EW (\alpha,\lambda_2,k)$  $(j=p+1,\ldots,n)$, and $Y_i \sim EW (\alpha,\mu_1,k)$  $(i=1,\ldots,p)$ and $Y_j \sim EW (\alpha,\mu_2,k )$  $(j=p+1,\ldots,n)$ be with a common Archimedean copula having generator $\psi$.  Then, if 
\begin{equation*}
(\alpha+\bar{\alpha}t)\ln\left(\frac{t}{\alpha+\bar{\alpha}t}\right)\left[\alpha+2\bar{\alpha}t-\frac{(\alpha+\bar{\alpha}
t)t\phi''(1-t)}{\phi'(1-t)}\right]
\end{equation*}
is increasing with respect to $t \in [0,1]$ and $0\le k\leq 1$, we have
\begin{eqnarray*}
(\lambda_1-\lambda_2)(\mu_1-\mu_2) \ge 0\,\,\,\,\,and\,\,\,\,\,\,
\frac{\lambda_{2:2}}{\lambda_{1:2}} \ge \frac{\mu_{2;2}}{\mu_{1:2}}
\Longrightarrow
Y_{n:n} \le_{Lorenz} X_{n:n},
\end{eqnarray*}
where $p+q=n$.
\end{corollary}

We now present some conditions for comparing the smallest order statistics in terms of dispersive order. In the following theorem, we use the notation
$I_+=\{(\lambda_{1},\ldots,\lambda_{n}):0<\lambda_{1}\leq\ldots\leq\lambda_{n}\}$ and
$D_+=\{(\lambda_{1},\ldots,\lambda_{n}):\lambda_{1}\geq\ldots\geq\lambda_{n}>0\} .$
\begin{theorem}
Let $X_i\sim EW(\alpha, \lambda_i,k)$ $(i=1,\ldots,n)$ and
$Y_i\sim EW(\alpha, \lambda ,k)$ $(i=1,\ldots,n)$ and the associated Archimedean survival copula for both be with generator $\psi,$ $0 \leq k \leq 1$ and $0 \leq \alpha \leq 1$. Then, $\lambda \leq (\lambda_1\cdots\lambda_n)^{\frac{1}{n}} $ implies $Y_{1:n} \succeq_{disp} X_{1:n}$ if ${\psi}/{\psi^{'}}$ is decreasing and concave.
\end{theorem}
\begin{proof}
First, let us consider the function 
$$F(x)= \frac{1-e^{-x^{k}}}{\left(1-\bar{\alpha} e^{-x^{k}}\right)}.$$
Let us define another function 
$$\bar{h}(x) = \frac{x f(x)}{\bar{F}(x)}. $$
Then,
\begin{equation*}
    \bar{h}(e^x) = \frac{e^x f(e^x)}{\bar{F}(e^x)} = \frac{ke^{xk}}{\left(1-\bar{\alpha} e^{-e^{xk}}\right)}.
\end{equation*}
Upon taking derivative of $\bar{h}(e^x)$ with respect to $x$, from Lemma \ref{dec-3}, we get
\begin{equation}\label{disp_1}
    \frac{\partial \bar{h}(e^x)}{\partial x}=  \dfrac{k^2\cdot\left(\mathrm{e}^{\mathrm{e}^{kx}}-\bar{\alpha}\mathrm{e}^{kx}-\bar{\alpha}\right)\mathrm{e}^{\mathrm{e}^{kx}+kx}}{\left(\mathrm{e}^{\mathrm{e}^{kx}}-\bar{\alpha}\right)^2} \geq 0.
\end{equation}
From \eqref{disp_1}, we have $\bar{h}(e^x)$ to be increasing in $x$ and so $h(x)$ is increasing in $x$. Again, the second-order partial derivative of $\bar{h}(e^x)$ with respect to $x$ is given by
\begin{equation}\label{disp_4}
    \frac{\partial^2 \bar{h}(e^x)}{\partial x^2} = \dfrac{k^3\cdot\left(\mathrm{e}^{2\mathrm{e}^{kx}}+\left(\bar{\alpha}\mathrm{e}^{2kx}-3\bar{\alpha}\mathrm{e}^{kx}-2\bar{\alpha}\right)\mathrm{e}^{\mathrm{e}^{kx}}+\bar{\alpha}^2\mathrm{e}^{2kx}+3\bar{\alpha}^2\mathrm{e}^{kx}+\bar{\alpha}^2\right)\mathrm{e}^{\mathrm{e}^{kx}+kx}}{\left(\mathrm{e}^{\mathrm{e}^{kx}}-\bar{\alpha}\right)^3} \geq 0,
\end{equation}
and so we see from Lemma \ref{dec-4} that $\bar{h}(e^x)$ is convex for all $x$. Now,
\begin{equation*}
\ln{\bar{F}(x)} = \log\left(1-\frac{\left(1-e^{-x^{k}}\right)}{\left(1-\bar{\alpha}e^{-x^{k}}\right)}\right)
\end{equation*}
whose second partial derivative is
\begin{equation}\label{disp_2}
\frac{\partial^2 \ln{\bar{F}(x)}}{\partial x^2} = -\dfrac{kx^{k-2}\mathrm{e}^{x^k}\cdot\left(\left(k-1\right)\mathrm{e}^{x^k}-\bar{\alpha}kx^k-\bar{\alpha}k+\bar{\alpha}\right)}{\left(\mathrm{e}^{x^k}-\bar{\alpha}\right)^2} \geq 0,
\end{equation}
where $0 \leq k \leq 1$ and $0 \leq \alpha \leq 1$.
Hence, from \eqref{disp_2}, we see that $\bar{F}(x)$ is log-convex in $x \geq 0$. Under the considered set-up, $X_{1:n}$ and $Y_{1:n}$ have their distribution functions as $F_1(x)=1-\psi(\sum_{k=1}^{n} \phi(\bar{F}(\lambda_k x)))$ and $H_1(x)=1-n\psi(n(\bar{F}(\lambda x)))$ for $x\geq 0$, and their density functions as 
\begin{equation*}
f_1(x)= \psi^{'}\Big(\sum_{k=1}^{n} \phi(\bar{F}(\lambda_k x))\Big)\sum_{k=1}^{n} \frac{\lambda_k h(\lambda_k x) \bar{F}(\lambda_k x)}{\psi^{'}(\phi(\bar{F}(\lambda_k x)))}
\end{equation*}
and
\begin{equation}\label{disp_5}
h_1(x)= \psi^{'}(n\phi(\bar{F}(\lambda x))) \frac{n\lambda h(\lambda x) \bar{F}(\lambda x)}{\psi^{'}(\phi(\bar{F}(\lambda x)))},
\end{equation}
respectively. Now, let us denote $L_1(x;\lambda)= \bar{F}^{-1}(\psi((1/n)\sum_{k=1}^{n} \phi(\bar{F}(\lambda_k x))))$. Then, for $x\geq0$, $H_1^{-1}(F_1(x))=(1/\lambda)L_1(x;\lambda)$ and 
\begin{equation}\label{eqn6}
h_1(H_1^{-1}(F_1(x)))= \psi^{'}\Big(\sum_{k=1}^{n} \phi(\bar{F}(\lambda_k x))\Big)\frac{n\lambda h(L_1(x;\lambda))\phi((1/n)\sum_{k=1}^{n} \phi(\bar{F}(\lambda_k x)))}{\psi^{'}\Big((1/n)\sum_{k=1}^{n} \phi(\bar{F}(\lambda_k x))\Big)}.
\end{equation}
{Again, concavity property of $\psi/\psi^{'}$ yields}
\begin{equation*}\label{eqn7}
\frac{\psi(\frac{1}{n} \sum_{k=1}^{n} \phi(\bar{F}(\lambda_k x)))}{\psi^{'}(\frac{1}{n} \sum_{k=1}^{n} \phi(\bar{F}(\lambda_k x)))}
\geq \frac{1}{n}\sum_{k=1}^{n}\frac{\psi(\phi(\bar{F}(\lambda_k x)))}{\psi^{'}(\phi(\bar{F}(\lambda_k x)))}.
\end{equation*}
As $h(x)$ is increasing and $\psi/\psi^{'}$ is decreasing, $\ln{\bar{F}(e^x)}$ is concave and $\ln{\psi}$ is convex. Now, using the given assumption that $\lambda\leq (\lambda_1\cdots\lambda_n)^{\frac{1}{n}}$, and the fact that $\ln{\bar{F}(x)}\leq 0$ is decreasing, we have

\begin{eqnarray*}
\ln{\bar{F}(\lambda x)}\geq \ln{\bar{F}((\Pi_{k=1}^{n}\lambda_k x)^{\frac{1}{n}})}\geq \frac{1}{n}\sum_{k=1}^{n} \ln{\bar{F}\left(exp{\left(\frac{1}{n}\sum_{k=1}^{n}\ln{(\lambda_k x)}\right)}\right)}.
\end{eqnarray*}
Observe that $\ln{\bar{F}(e^x)}$ is concave, $\ln{\psi}$ is convex, and {$\bm{\lambda} \in I_+ \text{ or } D_+$.}
Hence, from Chebychev's inequality, it follows that
\begin{align}\label{disp_3}
    \ln{\bar{F}(\lambda x)}-\ln{\psi\left(\frac{1}{n}\sum_{k=1}^{n}\phi\big(\bar{F}(\lambda_k x)\big)\right)}&\geq
\frac{1}{n}\sum_{k=1}^{n} \ln{\bar{F}\left(exp{(\frac{1}{n}\sum_{k=1}^{n}\ln{(\lambda_k x)})}\right)}-\ln{\psi\left(\frac{1}{n}\sum_{k=1}^{n}\phi\big(\bar{F}(\lambda_k x)\big)\right)}\nonumber\\
&\geq
\frac{1}{n}\sum_{k=1}^{n} \ln{\bar{F}(\lambda_k x)}-\frac{1}{n}\sum_{k=1}^{n} \ln{\bar{F}(\lambda_k x)}\nonumber\\&\geq 0.
\end{align}
So, from \eqref{disp_3}, we have $L_1(x;\lambda)\geq \lambda x$. Moreover, we have $h(x)$ to be decreasing as $\bar{h}(x)$ is increasing and so, $h(x)$ is convex. Therefore, using $\lambda \leq (\prod_{k=1}^{n}\lambda_k)^{\frac{1}{n}},$ we have
\begin{align*}
    \lambda h(L_1(x;\lambda))&\leq \frac{1}{x}\lambda x h(\lambda x) \nonumber\\
    &\leq \frac{1}{x}\big(\Pi_{k=1}^{n}\lambda_k x\big)^{\frac{1}{n}}. h\big(\big(\Pi_{k=1}^{n}\lambda_k x\big)^{\frac{1}{n}}\big)\nonumber\\
    &=\frac{1}{x}exp\left(\frac{1}{n}\sum_{k=1}^{n}\ln{\lambda_k} x\right).h\left(exp\left(\frac{1}{n}\sum_{k=1}^{n}\ln{\lambda_k x}\right)\right).
\end{align*}

Once again, by using Chebychev's inequality, increasing property of $\bar{h},$ decreasing property of ${\psi}/{{\psi}^{'}}$ and $\bm{\lambda} \in I_+\text{ or } D_+$, we obtain
\begin{align}\label{disp_6}
   \frac{1}{n}\sum_{k=1}^{n} \frac{\lambda_k h(\lambda_k x)\bar{F}(\lambda_k x)}{\psi^{'}(\phi(\bar{F}(\lambda_k x)))}&=\frac{1}{n}\sum_{k=1}^{n} \frac{\lambda_k h(\lambda_k x)\psi(\phi(\bar{F}(\lambda_k x)))}{\psi^{'}(\phi(\bar{F}(\lambda_k x)))}\nonumber\\
   &\leq \frac{1}{n}\sum_{k=1}^{n} \lambda_k h(\lambda_k x) \frac{1}{n}\sum_{k=1}^{n} \frac{ \psi(\phi(\bar{F}(\lambda_k x)))}{\psi^{'}(\phi(\bar{F}(\lambda_k x)))}\nonumber\\
   &\leq \frac{1}{n}\sum_{k=1}^{n} \lambda_k h(\lambda_k x) .\frac{ \psi\left(\frac{1}{n}\sum_{k=1}^{n} \phi(\bar{F}(\lambda_k x))\right)}{\psi^{'}\left(\frac{1}{n}\sum_{k=1}^{n} \phi(\bar{F}(\lambda_k x))\right)} \nonumber\\
   &\leq \lambda L_1(x;\lambda).\frac{ \psi(\frac{1}{n}\sum_{k=1}^{n} \phi\left(\bar{F}(\lambda_k x))\right)}{\psi^{'}\left(\frac{1}{n}\sum_{k=1}^{n} \phi(\bar{F}(\lambda_k x))\right)}.
\end{align}
Now, using the inequality in \eqref{disp_6}, \eqref{disp_5} and \eqref{eqn6}, we obtain, for all $x\geq 0$,
\begin{align*}
 & h_1(H_1^{-1}(F_1(x)))-f_1(x)\nonumber\\
  &=\frac{1}{n}\sum_{k=1}^{n} \frac{\lambda_k h(\lambda_k x)\bar{F}(\lambda_k x)}{\psi^{'}(\phi(\bar{F}(\lambda_k x)))}- \frac{1}{n}\sum_{k=1}^{n} \lambda_k h(\lambda_k x) .\frac{ \psi\left(\frac{1}{n}\sum_{k=1}^{n} \phi(\bar{F}(\lambda_k x))\right)}{\psi^{'}\left(\frac{1}{n}\sum_{k=1}^{n} \phi(\bar{F}(\lambda_k x))\right),}\nonumber\\
  &\leq 0, 
\end{align*}
which yields $f_1(F_1^{-1}(x))\geq h_1(H_1^{-1}(x)),$ for all $x\in(0,1).$ This completes the proof of the theorem.
\end{proof}
\subsection{Ordering results based on ramdom number of variables}\label{s22}

In this subsection, we will consider two sets of dependent $N_{1}$ and $N_{2}$ variables $\{X_{1},\ldots,X_{N_{1}}\}$ and $\{Y_{1},\ldots,Y_{N_{2}}\},$ where $X_{i}$ follows $EW(\alpha_{i}, \lambda_i,l_{i})$ and $Y_{i}$ follows $EW(\beta_{i}, \lambda_i,k_{i})$ coupled with Archimedean copulas having different generators. Under this set-up, we develop different ordering results based on the usual stochastic order, where in the model parameters are connected by different majorization orders. Here, the number of observations $N_{1}$ and $N_{2}$ are stochastically comparable, independent of $X_{i}'$s and $Y_{i}'$s, respectively. 
In the following theorem, if  $N_{1}\leq_{st}N_{2}$ and the tilt parameter $\alpha\in(0,\infty),$ then under the same conditions as in Theorem \ref{th1}, we can extend the corresponding results as follows,
\begin{theorem}\label{th4.1}
Let $X_i\sim EW(\alpha, \lambda_i,k)$, $(i=1,\ldots,n)$ and $Y_i\sim EW(\alpha, \mu_i,k)$ $(i=1,\ldots,n)$ and their associated Archimedean survival copulas be with generators $\psi_1$ and $\psi_2$, respectively. Let $N_1$ be a non-negative integer-valued random variable independently of $X_{i}'$s and $N_2$ be a non-negative integer-valued random variable independently of $Y_{i}'$s. Further, let $N_{1}\leq_{st} N_{2},$ $\phi_2\circ\psi_1$ be super-additive and $\psi_{1}$ be log-concave. Then, for $0<\alpha\leq 1$, we have 
{$$(\log{\lambda_1},\ldots,\log{\lambda_n})\succeq_{w}(\log{\mu_1}, \ldots, \log{\mu_n})\Rightarrow Y_{1:{N_{2}}}\succeq_{st}X_{1:{N_{1}}}.$$}
\end{theorem}
\begin{proof}
{ The survival functions of $X_{1:n}$ and $Y_{1:n}$ are given by }
\begin{equation*}
    \bar{F}_{X_{1:n}}(x) = \psi_{1}\Big[\sum_{m=1}^{n} \phi_{1}\Big(\frac{\alpha e^{-(x\lambda_m)^k}}{1-\bar{\alpha}e^{-(x\lambda_m)^k}}\Big)\Big]
\end{equation*}
{  and}
\begin{equation*}
    \bar{F}_{Y_{1:n}}(x) = \psi_{2}\Big[\sum_{m=1}^{n} \phi_{2}\Big(\frac{ \alpha e^{-(x\mu_m)^k}}{1-\bar{\alpha}e^{-(x\mu_m)^k}}\Big)\Big],
\end{equation*}
respectively.
As $N_{1}\leq_{st}N_{2},$ we get
\begin{align}
    \bar{F}_{X_{1:{N_{1}}}}(x) &=\sum_{m=1}^{n}P({X_{1:{N_{1}}}> x}|N_{1}=m)P(N_{1} =m)\nonumber\\
    &=\sum_{m=1}^{n}P({X_{1:{m}}> x})P(N_{1} =m)\nonumber\\
    &\leq\sum_{m=1}^{n}P({X_{1:{m}}> x})P(N_{2} =m)\nonumber\\
    &\leq\sum_{m=1}^{n}P({Y_{1:{m}}> x})P(N_{2} =m)\nonumber\\
    &\leq \bar{F}_{Y_{1:{N_{2}}}}(x)
\end{align}
by using Theorem \ref{th1}, as required.
\end{proof}
We similarly extend Theorem \ref{th2} by considering $N_{1}$ and $N_{2}$ to be random. The result provides us sufficient conditions for comparing two parallel systems, wherein the components lifetimes follow dependent extended Weibull family of distributions having different scale parameters.
\begin{theorem}\label{th4.2}
Let $X_i\sim EW(\alpha, \lambda_i,k)$ $(i=1,\ldots,n)$ and  $Y_i\sim EW(\alpha, \mu_i,k)$ $(i=1,\ldots,n),$ where $0<k\leq1$  and their associated Archimedean copulas be with generators $\psi_1$ and $\psi_2$, respectively. Let $N_1$ be a non-negative integer-valued random variable independently of $X_{i}'$s and $N_2$ be a non-negative integer-valued random variable independently of $Y_{i}'$s. Also, let $N_{1}\leq_{st}N_{2}$ and $\phi_2\circ\psi_1$ be super-additive. Then, for $0<\alpha\leq 1$, we have
$$\boldsymbol{\lambda}\succeq^{w}\boldsymbol{\mu}\Rightarrow X_{{N_{1}}:{N_1}}\succeq_{st}Y_{{N_{2}}:{N_2}}.$$
\end{theorem}

\begin{proof}
{ The distribution functions of $X_{1:n}$ and $Y_{1:n}$ are given by}
\begin{equation}
    F_{X_{n:n}}(x) = \psi_{1}\Big[\sum_{m=1}^{n} \phi_{1}\Big(\frac{1- e^{-(x\lambda_m)^k}}{1-\bar{\alpha}e^{-(x\lambda_m)^k}}\Big)\Big]
\end{equation}
{  and}
\begin{equation}
    F_{Y_{n:n}}(x) = \psi_{2}\Big[\sum_{m=1}^{n} \phi_{2}\Big(\frac{ 1- e^{-(x\mu_m)^k}}{1-\bar{\alpha}e^{-(x\mu_m)^k}}\Big)\Big],
\end{equation}
respectively.
As $N_{1}\leq_{st}N_{2},$ we get
\begin{align}
    F_{X_{{N_{1}}:{N_{1}}}}(x) &=\sum_{m=1}^{n}P({X_{{N_{1}}:{N_{1}}}< x}|N_{1}=m)P(N_{1} =m)\nonumber\\
    &=\sum_{m=1}^{n}P({X_{{m}:{m}}< x})P(N_{1} =m)\nonumber\\
    &\leq\sum_{m=1}^{n}P({X_{{m}:{m}}<x})P(N_{2} =m)\nonumber\\
    &\leq\sum_{m=1}^{n}P({Y_{{m}:{m}}< x})P(N_{2} =m)\nonumber\\
    &\leq F_{Y_{{N_{2}}:{N_{2}}}}(x)
\end{align}
by using Theorem \ref{th2}, as required.
\end{proof}
The following theorem is an extension of Theorem \ref{th4} for the case when $N_{1}$ and $N_{2}$ are random when the shape parameters $\bm{l}$ and $\bm{k}$ are connected in majorization order and the tilt parameters $\bm{\alpha}$ and $\bm{\beta}$ are equal, scalar-valued and lie between $0$ and $1.$
\begin{theorem}\label{th4.4}
Let $X_i\sim EW(\alpha, \lambda,k_i)$ $(i=1,\ldots,n)$ and $Y_i\sim EW(\alpha, \lambda,l_i)$ $(i=1,\ldots,n)$, and the associated Archimedean copulas be with generators $\psi_1$ and $\psi_2$, respectively. Let $N_1$ be a non-negative integer-valued random variable independently of $X_{i}'$s and $N_2$ be a non-negative integer-valued random variable independently of $Y_{i}'$s. Further, let $N_{1}\leq_{st}N_{2},$ $\phi_2\circ\psi_1$ be super-additive and $\alpha t\phi_1^{''}(t)+\phi_1^{'}(t)\geq 0$. Then, for $0<\alpha\leq 1$, we have
$$ \boldsymbol{l}\succeq^{m}\boldsymbol{k} \Rightarrow X_{{N_{1}}:{N_{1}}}\succeq_{st}Y_{{N_{2}}:{N_{2}}}.$$

\end{theorem}

The above result states that the survival function of $X_{{N_{1}}:{N_{1}}}$ is stochastically less than that of $Y_{{N_{2}}:{N_{2}}},$ whereas the following theorem presents conditions under which the survival function of $X_{{1}:{N_{1}}}$ is stochastically less than that of $Y_{1:{N_{2}}}.$ These can be proved using the same arguments as in the proofs of Theorems \ref{th4} and \ref{th4.1}.
\begin{theorem}\label{th4.5}
Let $X_i\sim EW(\alpha, \lambda,k_i)$ $(i=1,\ldots,n)$ and $Y_i\sim EW(\alpha, \mu,l_i)$ $(i=1,\ldots,n)$ and the associated Archimedean survival copulas be with generators $\psi_1$ and $\psi_2$, respectively.  Let $N_1$ be a non-negative integer-valued random variable independently of $X_{i}'$s and $N_2$ be a non-negative integer-valued random variable independently of $Y_{i}'$s. Further, let $N_{1}\leq_{st}N_{2},$ $\phi_2\circ\psi_1$ be super-additive and $t\phi_1^{''}(t)+\phi_1^{'}(t)\geq 0$. Then, for $0<\alpha\leq 1$, we have
$$ \boldsymbol{l}\succeq^{m}\boldsymbol{k} \Rightarrow Y_{1:{N_{2}}}\succeq_{st}X_{1:{N_{1}}}.$$
\end{theorem}
We finally present similar comaparison reults when the tilt parameter vectors are connected in majorization order and the parallel and series systems have extended Weibull components.  
\begin{theorem}\label{th4.6}
Let $X_i\sim EW(\alpha_i,\lambda,k)$ $(i=1,\ldots,n)$ and $Y_i\sim EW(\beta_i, \lambda, k)$ $(i=1,\ldots,n)$ and the associated Archimedean copulas be with generators $\psi_1$ and $\psi_2$, respectively. Let $N_1$ be a non-negative integer-valued random variable independently of $X_{i}'$s and $N_2$ be a non-negative integer-valued random variable independently of $Y_{i}'$s. Further, let $N_{1}\leq_{st}N_{2},$ $\phi_2\circ\psi_1$ be super-additive and 
\begin{equation*}
t\phi_1^{''}(t)+2\phi_1^{'}(t)\geq 0.
\end{equation*} Then, we have
$$ \boldsymbol{\alpha}\succeq_{w}\boldsymbol{\beta} \Rightarrow X_{{N_{1}}:{N_{1}}}\succeq_{st}Y_{{N_{2}}:{N_{2}}}.$$
\end{theorem}
\begin{theorem}\label{th4.7}
Let $X_i\sim EW(\alpha_i,\lambda,k)$ $(i=1,\ldots,n)$ and $Y_i\sim EW(\beta_i, \lambda, k)$ $(i=1,\ldots,n)$ and the associated Archimedean survival copulas be with generators $\psi_1$ and $\psi_2$, respectively. Let $N_1$ be a non-negative integer-valued random variable independently of $X_{i}'$s and $N_2$ be a non-negative integer-valued random variable independently of $Y_{i}'$s. Further, let $N_{1}\leq_{st}N_{2}$ and $\phi_2\circ\psi_1$ be super-additive. Then, we have
$$ \boldsymbol{\alpha}\succeq^{w}\boldsymbol{\beta} \Rightarrow X_{1:{N_{1}}}\succeq_{st}Y_{1:{N_{2}}}.$$
\end{theorem}
\section{Concluding remarks}\label{con}
The purpose of this article is to establish ordering results between two given sets of dependent variables $\{X_{1},\ldots,X_{n}\}$ and  $\{Y_{1},\ldots,Y_{n}\},$ wherein the underlying variables are from extended Weibull family of distributions. The random variables are then associated with Archimedean (survival) copulas with different generators. Several conditions are presented for the stochastic comparisons of extremes in the sense of usual stochastic, star, Lorenz, hazard rate, reversed hazard rate and dispersive orders. Further, we have derived inequalities between the extreme order statistics when the number of variables in the two sets are random satisfying the usual stochastic order. We have also presented several examples and counterexamples to illustrate all the established results and their implications.
\\\\
{\bf Disclosure statement}\\\\
No potential conflict of interest was reported by the authors.
\\\\
{\bf Funding}\\\\
SD thanks the Theoretical Statistics and Mathematics Unit, Indian Statistical Institute, Bangalore for the financial support and the hospitality during her stay.


\end{document}